\documentclass[12pt]{article}

\usepackage[english]{babel}

\usepackage[a4paper,top=2cm,bottom=2cm,left=3cm,right=3cm,marginparwidth=1.75cm]{geometry}

\usepackage{amsmath}
\usepackage{graphicx}
\usepackage{yfonts}
\usepackage{amssymb}
\usepackage{amsthm,bm}
\usepackage{comment}
\usepackage{bbm}
\usepackage{cite}
 \usepackage{tikz}
 \usepackage{anyfontsize}
 \usetikzlibrary{arrows}
\newtheorem{thm}{Theorem}[section]
\newtheorem{co}[thm]{Corollary}
\newtheorem{lem}[thm]{Lemma}

\newtheorem{assumption}[thm]{Assumption}

\newtheorem{pr}[thm]{Proposition}

\newtheorem{definition}[thm]{Definition}

\newtheorem{example}[thm]{Example}

\newtheorem{remark}[thm]{Remark}

\usepackage[colorlinks=true, allcolors=blue]{hyperref}

\newcommand{\E}{\mathbb{E}}
\newcommand{\EX}{{\mathbb{E}}}
\newcommand{\PX}{{\mathbb{P}}}

\title{Feedback Capacity of the Continuous-Time ARMA(1,1) Gaussian Channel}
\author{\footnotesize \begin{tabular}{ccc}
Jun Su & Guangyue Han & Shlomo Shamai (Shitz)\\
The University of Hong Kong &The University of Hong Kong& Technion-Israel Institute of Technology\\
 email:  junsu1995@connect.hku.hk& email:  ghan@hku.hk& email: sshlomo@ee.technion.ac.il\\
\end{tabular}}
\begin{document}
\maketitle
\renewcommand{\thefootnote}{\fnsymbol{footnote}}
\footnotetext{\hspace*{0mm} A preliminary version of this work has been presented in IEEE International Symposium on Information Theory (ISIT) 2023.}

\begin{abstract}
We consider the continuous-time ARMA(1,1) Gaussian channel and derive its feedback capacity in closed form. More specifically, the channel is given by $\boldsymbol{y}(t) =\boldsymbol{x}(t) +\boldsymbol{z}(t)$, where the channel input $\{\boldsymbol{x}(t) \}$ satisfies average power constraint $P$ and the noise $\{\boldsymbol{z}(t)\}$ is a first-order {\em autoregressive moving average} (ARMA(1,1)) Gaussian process satisfying
$$
\boldsymbol{z}^\prime(t)+\kappa \boldsymbol{z}(t)=(\kappa+\lambda)\boldsymbol{w}(t)+\boldsymbol{w}^\prime(t),
$$
where $\kappa>0,~\lambda\in\mathbb{R}$ and $\{\boldsymbol{w}(t) \}$ is a white Gaussian process with unit double-sided spectral density.

We show that the feedback capacity of this channel is equal to the unique positive root of the equation
$$
P(x+\kappa)^2 = 2x(x+\vert \kappa+\lambda\vert)^2
$$
when $-2\kappa<\lambda<0$ and is equal to $P/2$ otherwise. Among many others, this result shows that, as opposed to a discrete-time additive Gaussian channel, feedback may not increase the capacity of a continuous-time additive Gaussian channel even if the noise process is colored. The formula enables us to conduct a thorough analysis of the effect of feedback on the capacity for such a channel. We characterize when the feedback capacity equals or doubles the non-feedback capacity; moreover, we disprove continuous-time analogues of the half-bit bound and Cover's $2P$ conjecture for discrete-time additive Gaussian channels.
\end{abstract}

\section{Introduction}
We start with the following point-to-point continuous-time additive white Gaussian noise (AWGN) channel \cite{shannon1949communication}
\begin{equation}\label{AWGN1}
\boldsymbol{y}(t)=\boldsymbol{x}(t)+\boldsymbol{w}(t),~-\infty<t<+\infty,
\end{equation}
where the channel noise $\{\boldsymbol{w}(t) \}$ is a white Gaussian process with unit double-sided spectral density, $\{\boldsymbol{x}(t) \}$ is the channel input and $\{\boldsymbol{y}(t)\}$ is the channel output. Since $\{\boldsymbol{w}(t)\}$ can be regarded as the derivative $\{\dot{B}(t)\}$ of the standard Brownian motion $\{B(t)\}$ in the generalized sense \cite{koralov2007theory,gel2014generalized}, or equivalently, $\{B(t)\}$ is the ``integral'' of $\{\boldsymbol{w}(t)\}$, the AWGN channel as in (\ref{AWGN1}) can be alternatively characterized by
\begin{equation}\label{AWGN2}
Y(t)=\int_0^tX(u)du+B(t),~t\ge 0,
\end{equation}
where $X=\{X(t)\}$ is the channel input and $Y=\{Y(t) \}$ is the channel output. Unlike white Gaussian noise, which is a generalized stochastic process in the sense of Schwartz's distribution \cite{obata2006white}, Brownian motion is an ordinary stochastic process that has been extensively studied in stochastic calculus. Evidently, the two formulations as in (\ref{AWGN1}) and (\ref{AWGN2}) allow us to examine an AWGN channel from different perspectives; in particular, the use of Brownian motion equips us with a wide range of established tools and techniques in stochastic calculus (see, e.g., \cite{ihara1993information,kadota1971mutual,han2021sampling,liu2019continuous} and references therein).

This paper is concerned with the following point-to-point continuous-time additive colored Gaussian noise (ACGN) channel
\begin{equation}\label{ACGN1}
\boldsymbol{y}(t)=\boldsymbol{x}(t)+\boldsymbol{z}(t),~-\infty<t<+\infty,
\end{equation}
where the channel noise $\boldsymbol{z}=\{\boldsymbol{z}(t)\}$ is a (possibly colored and generalized) stationary Gaussian process. Evidently, AWGN channels are a degenerated case of ACGN channels. Similarly as above, the ACGN channel as in (\ref{ACGN1}) can be alternatively characterized by
\begin{equation}\label{OU-CTGC}
Y(t)=\int_0^t X(u)du+Z(t),~~t\ge 0,
\end{equation}
where $\{Z(t)\}$ is the ``integral'' of $\{\boldsymbol{z}(t)\}$. Following~\cite{ihara1993information}, the treatment of ACGN channels in this work is mainly based on the formulation in (\ref{OU-CTGC}).

For any $M \in \mathbb{N}$ and $T >0$, an {\it $(M, T)$ code} for the ACGN channel (\ref{OU-CTGC}) with average power constraint $P$ consists of the following:
\begin{itemize}
\item[($a$)] A message index $W$ independent of $\{Z(t);t\in[0,T]\}$ and uniformly distributed over $\{1,2,...,M\}$.
\item[($b$)] For the non-feedback case, an encoding function $g_u: \{1,2,..., M\} \rightarrow \mathbb{R}, ~u \in [0,T]$, yielding codewords  $X(u) = g_u(W)$;
           for the feedback case, an encoding function $g_u: \{1,2,...,M \} \times C[0, u] \rightarrow \mathbb{R}, ~u \in [0,T]$, yielding codewords  $X(u) = g_u(W, Y_0^{u-})$. For both cases, the classical average power constraint is satisfied:
           \begin{equation*} \label{AP-operational}
                \frac{1}{T}\int_0^T \mathbb{E}[\vert X(u) \vert^2] du \le P. 
            \end{equation*}

\item[($c$)] A decoding functional $\hat{g}: C[0, T] \rightarrow \{1,2,...,M\}$.
\end{itemize}
Here we remark that for the feedback case, as argued in \cite{ihara1993information}, the pathwise continuity of $\{Z(t)\}$ and therefore that of $\{Y(t)\}$ imply that $Y_0^{u-}$ in ($b$) can be replaced by $Y_0^u$, and so the channel output $\{Y(t)\}$ is in fact the unique solution to the following stochastic functional differential equation\footnote{In other words, the encoding function $g_t$ is chosen such that (\ref{sfde-1}) admits a unique solution $\{Y(t)\}$. It is worth noting that, for some special cases, e.g., $Z(t)=B(t)$, there are some explicit sufficient conditions on $g_t$ to ensure the existence and uniqueness of the solution to (\ref{sfde-1}) (see, e.g., \cite{ng2017relative,liu2019continuous,han2021sampling,karatzas1991brownian}).}:
\begin{equation} \label{sfde-1}
dY(t)=g_t(W, Y_0^{t})dt+dZ(t).
\end{equation}
The {\it error probability} $\pi^{(T)}$ for the $(M, T)$ code as above is defined as
$$
\pi^{(T)} = \text{P}(\hat{g}(Y_0^T) \neq W).
$$
A rate $R$ is {\it achievable} if, for all $T>0$, there exists $([e^{TR}],T)$ codes with $\lim_{T\to\infty}\pi^{(T)}=0$. The {\it channel capacity} is defined as the supremum of all achievable rates, denoted by $C_{nfb}(P)$ for the non-feedback case and $C_{fb}(P)$ for the feedback case.

The literature on continuous-time ACGN channels is vast, and so below we only survey those results that are most relevant to this work. It has been shown by Huang and Johnson~\cite{huang1962information,huang1963information} that $C_{nfb}(P)$ can be achieved by a Gaussian input. For a special family of ACGN channels, Hitsuda~\cite{hitsuda1975gaussian} has applied a canonical representation method to derive a fundamental formula for the channel mutual information (see Lemma \ref{mutual-info}); based on this result, Ihara \cite{ihara1980capacity} showed that $C_{fb}(P)$ can be achieved by a Gaussian input with an additive feedback term. Similarly as in the discrete-time case, the property that feedback can at most double the capacity of a continuous-time ACGN channel, i.e., $C_{fb}(P)\le 2C_{nfb}(P)$, is established by examining a discrete-time approximation of $\{Z(t)\}$ (see \cite{butman1969general,pinsker1968probability,ihara1990capacity}). Employing a Hilbert space approach~\cite{gelfand1959calculation,pinsker1964information}, Baker \cite{baker1987capacity,Baker1983models} has derived a theoretical formula for $C_{nfb}(P)$, which however is somewhat difficult to evaluate. When it comes to effective computation of $C_{nfb}(P)$ or $C_{fb}(P)$, to the best of our knowledge, there are only a few results featuring an ``explicit'' and ``computable'' formula, detailed below. Here, we remark that Baker, Ihara and Hitsuda have studied the capacity of some families of ACGN channels, yet under different types of power constraints (see \cite{Baker1983models,baker1987capacity,baker1991information,hitsuda1975gaussian}).

\begin{itemize}
\item[1.] For the ACGN channel formulated as in (\ref{ACGN1}), when $\boldsymbol{z}$ is a 
stationary Gaussian process with a rational and smooth spectrum, $C_{nfb}(P)$ can be determined by the water-filling method (see, e.g., \cite{fano1961transmission,pinsker1964information,gallager1968information,baker1991information}). More specifically,
\begin{equation*}\label{non-capacity-water-filling}
C_{nfb}(P)=\frac{1}{2} \int_{-\infty}^{\infty}\log\left[\max\left(\frac{A}{S_{\boldsymbol{z}}(x)},1\right)\right]dx,
\end{equation*}
where $S_{\boldsymbol{z}}(x)$ is the spectral density function of the noise process $\boldsymbol{z}$ and the water level $A$ is a constant determined by
\begin{equation*}\label{water-filling-level}
P=\int_{[S_{\boldsymbol{z}}(x)\le A]}(A-S_{\boldsymbol{z}}(x))dx.
\end{equation*}
\item[2.] For the AWGN channel as in (\ref{AWGN1}) or (\ref{AWGN2}), it is a classical result that $C_{nfb}(P)=P/2$ and feedback does not increase the channel capacity, that is to say, $C_{fb}(P)=P/2$ (see, e.g.,~\cite{ihara1993information,kadota1971mutual,shannon1948mathematical}). Moreover, $C_{fb}(P)$ can be achieved by an additive feedback coding scheme that maximizes the channel mutual information and minimizes the filtering error simultaneously \cite{ihara1973optimal,liptser1974optimal,ihara1974coding}.
\end{itemize}

In this paper, we will focus our attention on a special family of ACGN channels, the continuous-time ARMA(1,1) Gaussian channel, where the noise is a continuous-time ARMA(1,1) Gaussian process, as presented in Figure \ref{figure1}. Specifically, we define the continuous-time ARMA(1,1) Gaussian noise process with parameters $\lambda\in\mathbb{R}$ and $\kappa>0$ as 
\begin{equation}\label{Def-CARMA1}
\boldsymbol{z}(t) = \lambda \boldsymbol{u}(t) +\boldsymbol{w}(t),\quad -\infty < t < \infty,
\end{equation}
where, as before, $\boldsymbol{w}(t) = \dot{B}(t)$ is a white Gaussian process, and $\boldsymbol{u}(t)=\int_{-\infty}^te^{-\kappa(t-u)}dB(u)$ is a stationary Ornstein-Uhlenbeck (OU) process (see, e.g., \cite{brockwell2010carma}). Here, we remark that $\{B(t)\}$ has been extended to $(-\infty,0)$, a common practice in the context of the stationary OU process (see, e.g., \cite{maller2009ornstein,brockwell2009levy,brockwell2014recent}). It turns out that $\{\boldsymbol{z}(t) \}$ defined in (\ref{Def-CARMA1}) is strictly stationary and satisfies the stochastic differential equation
\begin{equation}\label{Def-CARMA1-SDE}
\boldsymbol{z}^\prime(t)+\kappa \boldsymbol{z}(t)=(\kappa+\lambda)\boldsymbol{w}(t)+\boldsymbol{w}^\prime(t),
\end{equation}
where the derivative operator is interpreted in the generalized sense (see \cite{brockwell2010carma}). The equation (\ref{Def-CARMA1-SDE}) is the natural continuous-time analogue of the first order linear difference equation used to define a discrete-time ARMA$(1,1)$ process (see, e.g., \cite{brockwell2009time,ihara1993information}). Continuous-time ARMA processes have been of great interest to physicists and engineers (see, e.g., \cite{fowler1967statistical}).

\begin{figure}[htbp!]\label{figure1}
\includegraphics[width=13cm]{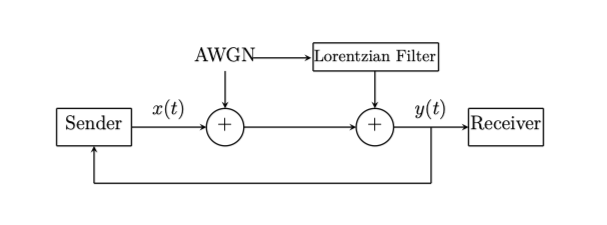}
\centering
\caption{The continuous-time ARMA$(1,1)$ Gaussian channel with feedback. Note that the OU process $\{\boldsymbol{u}(t)\}$ can be interpreted as the white Gaussian noise $\{\boldsymbol{w}(t) \}$ filtered by a Lorentzian filter (see, e.g., \cite{bibbona2008ornstein}).}
\end{figure}

\noindent Note that, similar to (\ref{OU-CTGC}), the continuous-time ARMA(1,1) Gaussian channel, as above, can be alternatively characterized by
\begin{equation}\label{OU-WGN}
Y(t)=\int_0^tX(u)du+B(t)+\lambda\int_0^t\int_{-\infty}^s e^{-\kappa(s-u)}dB(u)ds, \quad t \geq 0.
\end{equation}

Discrete-time ARMA Gaussian channels have been extensively studied (see, e.g., \cite{liu2018feedback,liu2017arma,han2016arma,kim2009feedback,kim2006feedback,ihara2019feedback,sabag2021feedback,gattami2018feedback,fang2021feedback}). However, to the best of our knowledge, there has been little progress on the feedback capacity for continuous-time ACGN channels. The main contribution in this work is an explicit characterization of the feedback capacity of the continuous-time ARMA$(1,1)$ Gaussian channel as in (\ref{OU-WGN}). Before this work, no ``explicit'' and ``computable'' formula is known for any nontrivial stationary ACGN channel (\ref{ACGN1}). Throughout the remainder of this paper, the notations $C_{nfb}(P)$ and $C_{fb}(P)$ will be reserved for the continuous-time ARMA$(1,1)$ Gaussian channel as in (\ref{OU-WGN}).

We will first derive a lower bound on $C_{fb}(P)$, which turns out to be tight for some cases. To achieve this, we will examine the following ACGN channel 
\begin{equation}\label{ACGN-B-1}
Y(t)=\int_0^tX(u)du+B(t)+\int_0^t\int_0^sh(s,u)dB(u)ds, \quad t\ge 0,
\end{equation}
where $h(s,u)$ is a Volterra kernel function on $L^2([0,T]^2)$ for all $T$. Here we emphasize that the channel (\ref{ACGN-B-1}) may not correspond to a stationary ACGN channel as in (\ref{ACGN1}). However, it can be shown that $\{B(t)+\int_0^t\int_0^sh(s,u)dB(u)ds\}$ is equivalent to the Brownian motion $\{B(t)\}$~\cite{hitsuda1968representation}, which renders the channel (\ref{ACGN-B-1}) more amenable to in-depth mathematical analysis, as evidenced by relevant results in the literature (see, e.g., \cite{hitsuda1975gaussian,ihara1980capacity,hitsuda1974mutual}).

More specifically, let $\{\Theta(t)\}$ be the message process, and let $\bar{I}_{\text{SK}}(\Theta; Y)$ denote the mutual information rate between $\{\Theta(t)\}$ and $\{Y(t)\}$ under the so-called continuous-time Schalkwijk-Kailath (SK) coding scheme (see, e.g., \cite{schalkwijk1966coding1,schalkwijk1968center,schalkwijk1966coding}). We will show (Theorem \ref{lower-bound-CTGC}) that
\begin{equation*} \label{ISK}
\bar{I}_{\text{SK}}(\Theta; Y)=Pr_P^2,
\end{equation*}
where $r_P$ is the limit of the unique solution to an ordinary differential equation, and moreover, one of the real roots of a third-order polynomial. It turns out that the continuous-time ARMA(1,1) Gaussian channel can be regarded as a special case of (\ref{ACGN-B-1}), and therefore $\bar{I}_{\text{SK}}(\Theta; Y)$ can help provide a lower bound on $C_{fb}(P)$.

With the aforementioned lower bound, we are ready to derive an explicit expression of $C_{fb}(P)$. More specifically, by examining discrete-time approximations of the channel (\ref{OU-WGN}), we prove (Theorem \ref{main-thm}) that for the case $-2\kappa<\lambda<0$, $C_{fb}(P)$ is upper bounded by $\bar{I}_{\text{SK}}(\Theta; Y)$, which means $C_{fb}(P) = \bar{I}_{\text{SK}}(\Theta; Y)$; for the other cases, we show $C_{fb}(P)=C_{nfb}(P)=P/2$. As a byproduct, this result shows that feedback may not increase the capacity of a continuous-time ACGN channel even if noise process is colored. By contrast, for a discrete-time ACGN channel, feedback does not increase the capacity if and only if the noise spectrum is white (see \cite[Corollary 4.3]{kim2009feedback}). 

Finally, we will devote our efforts to examining the effect of feedback on the capacity of the continuous-time ARMA$(1,1)$ Gaussian channel. It is well-known that feedback will increase yet can at most double the capacity of a continuous-time ACGN channel. We characterize (Theorem \ref{thm:C0-Cf}) when $C_{fb}(P)=C_{nfb}(P)$, $C_{nfb}(P)<C_{fb}(P)<2C_{nfb}(P)$ and $C_{fb}=2C_{nfb}(P)$, respectively. We further disprove the continuous-time analogues of some classical results and conjectures for discrete-time ACGN channels.

The remainder of the paper is organized as follows. In Section \ref{section-pre}, we introduce necessary notation and terminologies. We review the coding theorem for the feedback capacity and introduce the continuous-time SK coding scheme in Section \ref{sec:sk-coding-scheme}. Section \ref{lower-bound} provides an asymptotic characterization of $\overline{I}_{\text{SK}}(\Theta;Y)$ for a subclass of ACGN channels, which gives a lower bound on $C_{fb}(P)$. In Section \ref{proof-main-thm}, we derive an explicit formula for $C_{fb}(P)$. In Section \ref{sec:C0-Cf}, we characterize how feedback affects the capacity for the continuous-time ARMA$(1,1)$ Gaussian channel and further disprove continuous-time analogues of the half-bit bound and Cover's $2P$ conjecture for discrete-time ACGN channels.

\section{Notation and Terminologies} \label{section-pre}

We use $(\Omega,\mathcal{F},\PX)$ to denote the underlying probability space, and $\EX$ to denote the expectation with respect to the probability measure $\PX$. As is typical in the theory of stochastic calculus, we assume the probability space is equipped with a filtration $\{\mathcal{F}_t: 0 \leq t < \infty\}$, which satisfies the {\em usual conditions} \cite{karatzas1991brownian} and is rich enough to accommodate a standard Brownian motion.

Let $C[0, \infty)$ denote the space of all continuous functions over $[0, \infty)$, and let $C^1[0, \infty)$ be the space of all functions in $C[0, \infty)$ that have continuous derivatives on $[0, \infty)$. For $T > 0$, let $C[0, T]$ denote the space of all continuous functions over $[0, T]$. Let $X, Y$ be random variables defined on the probability space $(\Omega,\mathcal{F},\PX)$, which will be used to illustrate most of the notions and facts in this section (the same notations may have different connotations in other sections). Note that in this paper, a random variable can be real-valued with a probability density function, or path-valued (more precisely, $C[0, \infty)$- or $C[0, T]$-valued).

For any two probability measures $\mu$ and $\nu$, we write $\mu\sim\nu$ to mean they are equivalent, namely, $\mu$ is absolutely continuous with respect to $\nu$ and vice versa. For any two path-valued random variables $X_0^T = \{X(t); 0 \leq t \leq T\}$ and $Y_0^T =\{Y(t); 0 \leq t \leq T\}$, we use $\mu_{X_0^T}$ and $\mu_{Y_0^T}$ to denote the probability distributions on $C[0, T]$ induced by $X_0^T$ and $Y_0^T$, respectively, and $\mu_{X_0^T} \times \mu_{Y_0^T}$ the product distribution of $\mu_{X_0^T}$ and $\mu_{Y_0^T}$; moreover, we will use $\mu_{X_0^T, Y_0^T}$ to denote their joint probability distribution on $C[0, T] \times C[0, T]$. We say that $X_0^T$ is equivalent to $Y_0^T$ if $\mu_{X_0^T}\sim\mu_{Y_0^T}$, and moreover, $\{X(t)\}$ is equivalent to 
$\{Y(t)\}$ if $X_0^T$ is equivalent to $Y_0^T$ for all $T$. Besides, we use $\mathcal{F}_T(Y)$ to denote the $\sigma$-field generated by $Y_0^T$.

By Hitsuda \cite{hitsuda1968representation}, if a Gaussian process $\{Z(t)\}$ is equivalent to a given Brownian motion, then there exists a (possibly different) Brownian motion $\{ B(t)\}$ such that $Z(t)$ can be uniquely represented by
\begin{equation}\label{Gaussian-e-Brownian}
    Z(t)=B(t)+\int_0^t\int_0^s h(s,u)d B(u)ds,
\end{equation}
where $h(s,u)$ is a Volterra kernel function on $L^2([0,T]^2)$ for all $T$, i.e., $h(s,u)=0$ if $s<u$ and $\int_0^T\int_0^Th(s,u)^2dsdu<\infty$ for all $T$. Conversely, for a given Brownian motion $\{B(t)\}$, if $\{Z(t)\}$ has a representation in the form (\ref{Gaussian-e-Brownian}), then $\{Z(t)\}$ is equivalent to $\{B(t)\}$. Note that there exists a Volterra kernel function $l(s,u)\in L^2([0,T]^2)$ for all $T$, referred to as the {\em resolvent kernel} of $h(s,u)$, such that
\begin{equation}\label{resolvent-kernel}
\begin{aligned}
-h(s,u)&=l(s,u)+\int_u^sh(s,v)l(v,u)dv\\
&=l(s,u)+\int_u^sl(s,v)h(v,u)dv
\end{aligned}
\end{equation}
for all $s,u\in[0,\infty)$ (see \cite[Chapter 2]{smithies1958integral}). Therefore, the Brownian motion $\{B(t)\}$ can be also uniquely determined in terms of $\{Z(t)\}$ as
\begin{equation}\label{l-kernel}
B(t)=Z(t)+\int_0^t\int_0^sl(s,u)dZ(u)ds.
\end{equation}

The {\em mutual information} $I(X;Y)$ between two real-valued random variables $X,Y$ is defined as
\begin{equation} \label{def:mutual-information-rv}
I(X; Y) = \EX\left[\log \frac{f_{X, Y}(X, Y)}{f_X(X) f_Y(Y)}\right],
\end{equation}
where $f_X, f_Y$ denote the probability density functions of $X, Y$, respectively, and $f_{X,Y}$ their joint probability density function. More generally, for two $C[0, T]$-valued random variables $X_0^T, Y_0^T$, we define
{\scriptsize \begin{equation} \label{definition-mutual-information}
I(X_0^T; Y_0^T)=\begin{cases}
\EX\left[\log \frac{d \mu_{X_0^T, Y_0^T}}{d \mu_{X_0^T} \times \mu_{Y_0^T}}(X_0^T, Y_0^T)\right], & \mbox{ if } \frac{d \mu_{X_0^T, {Y_0^T}}}{d \mu_{X_0^T} \times \mu_{Y_0^T}} \mbox{ exists},\\
\infty, & \mbox{ otherwise},
\end{cases}
\end{equation}}\\
where $\frac{d \mu_{X_0^T, {Y_0^T}}}{d \mu_{X_0^T} \times \mu_{Y_0^T}}$ denotes the Radon-Nikodym derivative of $\mu_{X_0^T, {Y_0^T}}$ with respect to $\mu_{X_0^T} \times \mu_{Y_0^T}$.

The notion of mutual information can be further extended to generalized random processes, which we will only briefly describe and we refer the reader to~\cite{gelfand1959calculation} for a more comprehensive exposition.

Let $\mathcal{D}$ be the space of {\it test functions} over $\mathbb{R}$, i.e., all infinitely differentiable real functions with bounded support.
~The {\em mutual information} between two generalized random processes $\boldsymbol{x}=\{\boldsymbol{x}(\phi);\phi\in\mathcal{D} \}$ and $\boldsymbol{y}=\{\boldsymbol{y}(\psi);\psi\in\mathcal{D}\}$ is defined as
\begin{equation} \label{general-definition}
I(\boldsymbol{x}; \boldsymbol{y}) = \sup I(\boldsymbol{x}(\phi_1), \boldsymbol{x}(\phi_2), \dots, \boldsymbol{x}(\phi_m); \boldsymbol{y}(\psi_1), \boldsymbol{y}(\psi_2), \dots, \boldsymbol{y}(\psi_n)),
\end{equation}
where the supremum is over all possible $n, m \in \mathbb{N}$ and all possible test functions $\phi_1, \phi_2, \dots, \phi_m$, $\psi_1, \psi_2, \dots, \psi_n \in\mathcal{D}$. It is well-known that any ordinary stochastic process with locally integrable paths, $\{V(t);t\in\mathbb{R} \}$, will correspond to a generalized stochastic process $\{\boldsymbol{v}(\phi);\phi\in\mathcal{D} \}$ defined by 
$$
\boldsymbol{v}(\phi)=\int_\mathbb{R} V(t) \phi(t) dt.
$$
Then, it can be verified that the general definition of mutual information as in (\ref{general-definition}) includes (\ref{def:mutual-information-rv}) and (\ref{definition-mutual-information}) as special cases; moreover, when one of $\boldsymbol{x}$ and $\boldsymbol{y}$, say, $\boldsymbol{y}$, is a random variable, the general definition boils down to
$$
I(\boldsymbol{x}; \boldsymbol{y}) = \sup I(\boldsymbol{x}(\phi_1), \boldsymbol{x}(\phi_2), \dots, \boldsymbol{x}(\phi_m); \boldsymbol{y}),
$$
where the supremum is over all possible $m \in \mathbb{N}$ and all possible test functions $\phi_1, \phi_2, \dots, \phi_m\in\mathcal{D}$.

\section{Continuous-Time SK Coding}\label{sec:sk-coding-scheme}
In this section, we shall examine the continuous-time ACGN channel (\ref{ACGN-B-1}). Throughout this section, let $Z(t)=B(t)+\int_0^t\int_0^sh(s,u)dB(u)ds$.

The celebrated channel coding theorem by Shannon \cite{shannon1948mathematical} states, roughly speaking, that for a discrete memoryless channel, the capacity can be written as a supremum of the mutual information between the channel input and output. This classical result has been extensively extended and generalized to various channel models. Not surprisingly, under some mild assumptions, similar results hold for the non-feedback and feedback capacity of our channel. We will present a coding theorem for the feedback capacity below, i.e., \cite[Theorem 1]{Ihara1992channelcodingthm}, while that for the non-feedback capacity can be found in Section \ref{sec:proof-of-lower-bound}.

For the purpose of presenting the coding theorem, instead of transmitting a message index $W$, a random variable taking values from a finite alphabet, we will transmit a message process $\Theta = \{\Theta(t)\}$, a real-valued random process. Then, compared to (\ref{sfde-1}), the associated stochastic functional differential equation will take the following form:
\begin{equation*} \label{sfde-2}
dY(t)=g_t(\Theta(t), Y_0^{t})dt+dZ(t),
\end{equation*}
where we have set 
\begin{equation}\label{coding-thm-input}
X(t)=g_t(\Theta(t), Y_0^{t}).
\end{equation}
Following \cite{ihara1993information}, we consider the so-called {\em $T$-block feedback capacity} 
\[
C_{fb,T}(P) =\sup_{(\Theta,X)}\frac{1}{T}I(\Theta_0^T;Y_0^T)
\]
where the supremum is taken over all pairs $(\Theta,X)$ satisfying the following constraint
\begin{equation}\label{AP-T}
\frac{1}{T}\int_0^T \E [X^2(t)]dt\leq P.
\end{equation}
Now, we define 
$$
\bar{I}(\Theta;Y)=\limsup_{T\to\infty}\frac{1}{T}I(\Theta_0^T; Y_0^T),
$$
provided the limit exists, and further define
\[
C_{fb,\infty}(P)=\sup_{(\Theta,X)}\overline{I}(\Theta;Y),
\]
where the supremum is taken for all pairs $(\Theta,X)$ satisfying the constraint
\begin{equation}\label{AP-infinite}
\varlimsup_{T\to\infty} \frac{1}{T}\int_0^T\E[X^2(t)]dt\le P.
\end{equation}
Then, under some regularity conditions, the following coding theorem states that the feedback capacity of the channel (\ref{ACGN-B-1}) is equal to $C_{fb,\infty}(P)$.
\begin{thm}[\!\!{\cite[Theorem 1]{Ihara1992channelcodingthm}}]\label{coding-theorem}
Assume that
\begin{equation*}
\lim_{T\to\infty}\frac{1}{T}C_{fb,T}(P)=0.
\end{equation*}
If $R< C_{fb,\infty}(P)$ and $C_{fb,\infty}(P)$ is continuous in $P$, then the rate $R$ is achievable. Conversely, if a rate $R$ is achievable, then $R\le C_{fb,\infty}(P)$.
\end{thm}

The following lemma characterizes the relationship between the mutual information and the causal minimum mean-square error (CMMSE), generalizing the classical I-CMMSE relationship in \cite{kadota1971mutual,duncan1970calculation}.

\begin{lem}[\!\!{\cite[Theorem 1]{ihara1990capacity}}]\label{mutual-info}
Suppose $\int_0^T\E[X^2(t)]dt< \infty$. Then, we have
\begin{equation*}
I(\Theta_0^T;Y_0^T)=\frac{1}{2}\int_0^T \E\left[\left\vert X_l(t)-\E[X_l(t)|\mathcal{F}_t(Y)]\right\vert^2\right]dt,
\end{equation*}
where 
\begin{equation*}\label{l-op}
X_l(t)=X(t)+\int_0^tl(t,u)X(u)du,
\end{equation*}
and $l=l(s,u)$ is the resolvent kernel of $h$ in $L^2([0,T]^2)$.
\end{lem}
When it comes to the $T$-block feedback capacity of the channel (\ref{ACGN-B-1}), the so-called additive feedback coding scheme can achieve $C_{fb,T}(P)$ (see, e.g., \cite{ebert1970capacity,ihara1980capacity}). This coding scheme is formulated as follows. Consider the additive feedback coding scheme $(\Theta,X)=(\{\Theta(t)\},\{X(t)\})$ with $X(t)=\Theta(t)-\zeta(t)$, where $\zeta=\{\zeta(t)\}$ represents the feedback term, causally dependent on the output $Y=\{ Y(t)\}$, and is appropriately chosen such that the stochastic functional differential equation
\begin{equation}\label{additive-feedback}
Y(t)=\int_0^t \big(\Theta(s)-\zeta(s)\big) ds+Z(t)
\end{equation}
admits a unique solution.
~Obviously, if the feedback term $\zeta$ vanishes and so there is no feedback, (\ref{additive-feedback}) becomes
\[
Y^\ast (t)=\int_0^t \Theta(s)ds+Z(t).
\]
For such additive feedback coding schemes, we have the following lemma, which slightly extends \cite[Theorem 6.2.3]{ihara1993information}.
\begin{lem}\label{mutual-filter}
Suppose that 
\[
\int_0^T\E[\Theta^2(t)]dt< \infty, ~~~\int_0^T\E[\zeta^2(t)]dt< \infty.
\]
Then, for any $t\in[0,T]$, we have  
\begin{align}\label{sigma-algebra-no-change}
    I(\Theta_0^t;Y_0^t)&=I(\Theta_0^t;Y^{\ast,t}_0).
\end{align}
\end{lem}
Note that (\ref{sigma-algebra-no-change}) means that for the channel (\ref{ACGN-B-1}) under the additive feedback coding scheme, additive feedback will not provide the receiver with any new information. However, feedback can be used as a means to save transmission energy, since, for a fixed message $\Theta$, we can lower $\E[\vert\Theta(t)-\zeta(t)\vert^2]$ by appropriately choosing $\zeta$. This observation suggests an effective way to design a coding scheme to maximize $I(\Theta_0^T;Y_0^T)$ under the constraint (\ref{AP-T}). Indeed, the following result has been essentially proved by Ihara in \cite{ihara1980capacity}, for which a relatively more complete proof is provided in Appendix \ref{lemma2}.

\begin{thm}[\!\!{\cite[Theorem 3]{ihara1980capacity}} Reformulated]\label{detailed-structure}
For the continuous-time ACGN channel (\ref{ACGN-B-1}) under the constraint (\ref{AP-T}), $C_{fb,T}(P)$ can be achieved by a Gaussian pair $(\Theta,X)$ of the following form
\begin{equation}\label{opt-pair}
X(t)=\Theta(t)-\E[\Theta(t)\vert \mathcal{F}_t(Y^\ast)],~~ t\in[0,T],
\end{equation}
where 
$$
Y^\ast(t)=\int_0^t\Theta(s)ds+Z(t).
$$
Moreover, $\mathcal{F}_t(Y^\ast)=\mathcal{F}_t(Y)$, and so the pair $(\Theta,X)$ characterizes an additive feedback coding scheme of the form (\ref{additive-feedback}) where $\zeta(t)=\E[\Theta(t)\vert \mathcal{F}_t(Y)]$.
\end{thm}
The essence of the above theorem is that we can restrict our attention to the coding schemes of the form as in (\ref{opt-pair}). Following the spirits of the classical Schalkwijk-Kailath (SK) coding scheme, we formulate in our notation the continuous-time version of the celebrated SK coding scheme $(\Theta,X)$ in the form of
\begin{equation}\label{definition:CSK}
\begin{aligned}
X(t)&=\Theta(t)-\zeta(t)\\
&=A(t)\Theta_0-A(t)\E[\Theta_0\vert \mathcal{F}_t(Y^\ast)]
\end{aligned}
\end{equation}
satisfying 
$$
\E[X^2(t)]=P \quad \mbox{for any $t\ge 0$},
$$
where $\Theta_0$ is a standard Gaussian random variable and $A(t)$ is some deterministic function.

In general, the above continuous-time SK coding scheme can be invalid in the sense that $A(t)$ may not exist. However, in Sections \ref{lower-bound} and \ref{proof-main-thm}, we will show that the continuous-time SK coding scheme is valid for some special families of ACGN channels (\ref{ACGN-B-1}) and further is optimal for the continuous-time ARMA$(1,1)$ Gaussian channel in some cases.

\section{Mutual Information Rate}\label{lower-bound}
In this section, we narrow our attention to the special family of ACGN channels (\ref{ACGN-B-1}) in which the resolvent kernel $l(t,s)$ of $h(t,s)$ can be written as
\begin{equation}\label{hitsuda-separable}
l(t,s)=\frac{l_u(s)}{l_d(t)},\quad \mbox{for $t\ge s$},
\end{equation} 
where $l_u(t)\in C[0,+\infty)$ and $l_d(t)\in C^1[0,+\infty)$.

 We first state a lemma whose proof has been deferred to Appendix \ref{appendix-proof-of-ode}, which characterizes the asymptotics of the solution $g$ to the following ordinary differential equation (ODE):
\begin{equation}
\label{ode}
\left\{
\begin{aligned}
 g^\prime(t)&=-Pg^3(t)+\frac{P}{\sqrt{2}}g^2(t)+p(t)g(t)+\frac{1}{\sqrt{2}}q(t)
,\\
g(0)&=\frac{1}{\sqrt{2}},
\end{aligned}
\right.
\end{equation}
where $p(t),q(t)\in C[0,\infty)$ and $\lim_{t\to\infty}p(t)$ and $\lim_{t\to\infty}q(t)$ exist, denoted by $p$ and $q$, respectively.

\begin{lem}\label{exist-limit-ode}
For any $P$, the ODE (\ref{ode}) admits a unique solution $g(t)\in C^1[0,\infty)$. Moreover, $\lim_{t\to\infty}g(t)$ exists, which is one of the real roots of the following cubic equation:
$$
-Py^3+\frac{P}{\sqrt{2}}y^2+p y+\frac{q}{\sqrt{2}}=0.
$$
\end{lem}

Equipped with Lemma \ref{exist-limit-ode}, we can prove the following theorem.
\begin{thm}\label{lower-bound-CTGC}
Assume the resolvent kernel $l(t,s)$ of $h(t,s)$ in (\ref{ACGN-B-1}) can be written in the form (\ref{hitsuda-separable}) with 
\begin{equation}\label{lower-bound-thm:condition-rate}
\lim_{t\to\infty}\frac{l_u(t)}{l_d(t)}=\alpha,~\lim_{t\to\infty} \frac{l_d^\prime(t)}{l_d(t)}=\beta,
\end{equation}
where $\alpha,\beta\in\mathbb{R}$. Then, 
we have
\begin{equation}\label{I-SK}
\bar{I}_{SK}(\Theta;Y) =Pr_P^2,
\end{equation}
where $r_P=\lim_{t\to\infty}g(t)$ and $g$ is the solution of the ODE (\ref{ode}) with $p(t)=-l^\prime_d(t)/l_d(t)$ and $q(t)=(l_u(t)+l^\prime_d(t))/l_d(t)$. Moreover, $r_P$ is one of the real roots of the following cubic equation
\begin{equation}\label{3rd-poly}
    -Py^3+\frac{P}{\sqrt{2}}y^2-\beta y+\frac{\beta+\alpha}{\sqrt{2}}=0.
    \end{equation}
\end{thm}

\begin{proof}
Let $A(t)$ be a function defined by 
\begin{equation}\label{proof-lower-bound-A}
A(t) = \sqrt{P}e^{\int_0^tPg^2(s)ds}, ~~~t\ge 0,
\end{equation}
where the function $g(t)$ is defined to be a solution of the following Abel equation of the first kind (see, e.g., \cite{markakis2009closed}):
\begin{equation}\label{ode-Z}
\left\{
\begin{aligned}
 g^\prime(t)&=-Pg^3(t)+\frac{P}{\sqrt{2}}g^2(t)-\frac{l_d^\prime(t)}{l_d(t)}g(t)+\frac{1}{\sqrt{2}}\frac{l_d^\prime(t)+l_u(t)}{l_d(t)}
,\\
g(0)&=\frac{1}{\sqrt{2}}.
\end{aligned}
\right.
\end{equation}
It then follows from (\ref{lower-bound-thm:condition-rate}) and Lemma \ref{exist-limit-ode} that $\lim_{t\to\infty}g(t)$ exists, and the limit, denoted by $r_P$, is one of the real roots of the cubic equation (\ref{3rd-poly}). 

Next, we shall prove that the continuous-time SK coding scheme defined by (\ref{proof-lower-bound-A}) and (\ref{definition:CSK}) is valid, that is, for any $t\ge 0$,
\begin{equation}\label{proof-lower-bound-veri-CSK}
A^2(t)\E[\vert \Theta_0-\E[\Theta_0\vert \mathcal{F}_t(Y^\ast)]  \vert^2]=P.
\end{equation}
Indeed, since $g$ satisfies (\ref{ode-Z}), it holds that for any $t\ge 0$,
\begin{equation}\label{41}
\sqrt{2}g^\prime(t)+\sqrt{2}Pg^3(t)+\frac{\sqrt{2}l_d^\prime(t)}{l_d(t)}g(t)=Pg^2(t)+\frac{l_d^\prime(t)+l_u(t)}{l_d(t)}.
\end{equation}
Multiplying both sides of (\ref{41}) by $l_d(t)A(t)$, we obtain
\begin{equation*}
\begin{aligned}
\bigg(\sqrt{2}g^\prime(t)+\sqrt{2}Pg^3(t)&+\frac{\sqrt{2}l_d^\prime(t)}{l_d(t)}g(t)\bigg)l_d(t)A(t)\\
&=\sqrt{2P}\left(g^\prime(t)l_d(t)+l^\prime_d(t)g(t)+Pg^3(t)l_d(t)\right)e^{\int_0^t Pg^2(s)ds}\\
&=\sqrt{2}\frac{d}{dt}\left( g(t)l_d(t)A(t) \right)
\end{aligned}
\end{equation*}
and
\begin{equation*}
\begin{aligned}
\bigg(Pg^2(t)+\frac{l_d^\prime(t)+l_u(t)}{l_d(t)}\bigg)l_d(t)A(t)&=Pg^2(t)l_d(t)A(t)+l_d^\prime(t)A(t)+l_u(t)A(t)\\
&=\frac{d}{dt}\left(A(t)l_d(t)+\int_0^t l_u(s)A(s)ds   \right).
\end{aligned}
\end{equation*}
Therefore, (\ref{41}) leads to
\begin{equation}\label{choice-of-ld}
\sqrt{2}g(t)A(t)l_d(t)=l_d(t)A(t)+\int_0^tl_u(s)A(s)ds,
\end{equation}
which implies that
\begin{equation}\label{proof-lower-bound-AH}
2A(t)A^\prime(t)=PA_{l}^2(t),
\end{equation}
where
\begin{equation}\label{AL-A}
A_l(t) = A(t)+\int_0^t \frac{l_u(s)}{l_d(t)}A(s)ds.
\end{equation}
Therefore, using the initial condition $A^2(0)=P$, we have
\begin{equation}\label{A2-H}
A^2(t)=P(1+\int_0^tA_l^2(s)ds).
\end{equation}
On the other hand, it follows from (\ref{resolvent-kernel}) and (\ref{AL-A}) that
\begin{equation}\label{A-AL}
A(t) = A_l(t)+\int_0^t h(t,s)A_l(s)ds.
\end{equation}
Define 
\begin{equation}\label{Y-tilde-star}
\widetilde{Y}^\ast(t) = \int_0^tA_l(u)\Theta_0du +B(t).
\end{equation}
Then, we have
\begin{equation*}
\begin{aligned}
    Y^\ast(t) &= \int_0^tA(u)\Theta_0du+Z(t)\\
    &=\int_0^tA(u)\Theta_0du+B(t) +\int_0^t\int_0^sh(s,u)dB(u)ds\\
    &\overset{(a)}{=}\int_0^tH(t,u)A_l(u)\Theta_0du +\int_0^tH(t,u)dB(u) \\
    &=\int_0^tH(t,u)d\widetilde{Y}^\ast(u),
\end{aligned}
\end{equation*}
where $H(t,u)$ is the Volterra kernel function satisfying $H(t,u)=1+\int_u^th(s,u)ds$ for all $t\ge u$ and $(a)$ follows from (\ref{A-AL}). Hence, $\{ Y^\ast(s); s \le t \}$ and  $\{ \widetilde{Y}^\ast(s); s \le t \}$ are uniquely determined by each other and therefore, for any $t$
\begin{equation}\label{equi-2Y_star}
\mathcal{F}_t (Y^\ast) = \mathcal{F}_t (\widetilde{Y}^\ast).
\end{equation}
Now, applying \cite[Theorem 12.2]{liptser2013statistics} to the process $\{\widetilde{Y}^\ast(t)\}$ as in (\ref{Y-tilde-star}), we can readily establish 
\begin{equation}\label{filtering-wgc}
\E[\vert \Theta_0-\E[\Theta_0\vert \mathcal{F}_t(\widetilde{Y}^\ast)]  \vert^2]=\left(1+\int_0^tA_l^2(s)ds\right)^{-1},
\end{equation}
which, together with (\ref{equi-2Y_star}) and (\ref{A2-H}), immediately implies (\ref{proof-lower-bound-veri-CSK}), as desired. 

Finally, we are ready to prove (\ref{I-SK}). From Lemma \ref{mutual-info}, (\ref{equi-2Y_star}) and (\ref{filtering-wgc}), it follows that for a fixed $T$,
$$
I(\Theta_0^T;Y_0^T)=\frac{1}{2}\int_0^T \frac{A_l^2(t)}{1+\int_0^tA_l^2(s)ds}dt.
$$
Thus, we have 
\begin{equation*}
\begin{aligned}
\overline{I}_\text{SK}(\Theta;Y)&=\lim_{T\to\infty}\frac{1}{T}I(\Theta_0^T;Y_0^T)\\
&=\lim_{T\to\infty}\frac{1}{2T}\int_0^T \frac{A_l^2(t)}{1+\int_0^tA_l^2(s)ds}dt\\
&\overset{(b)}{=}\lim_{T\to\infty}\frac{1}{T}\int_0^T\frac{A^\prime(t)}{A(t)}dt\\
&=\lim_{T\to\infty}\frac{A^\prime(T)}{A(T)}\\
&=P\lim_{T\to\infty}g^2(T)\\
&\overset{(c)}{=}Pr_P^2,
\end{aligned}
\end{equation*}
where $(b)$ follows from (\ref{proof-lower-bound-AH}) and (\ref{A2-H}), and $(c)$ follows from Lemma \ref{exist-limit-ode}. Thus, (\ref{I-SK}) is established and then the proof is complete.
\end{proof}

\begin{remark}
It turns out that from the proof of Theorem \ref{lower-bound-CTGC}, $r_P$ is uniquely determined by $l(t,s)$, rather than the choice of $l_u(s)$ and $l_d(t)$.
\end{remark}

To illustrate possible applications of the above theorem, we give the following two examples.
\begin{example}\label{example-awgn}
{ \rm When $l(t,s)\equiv 0$, the channel (\ref{ACGN-B-1}) boils down to the AWGN channel (\ref{AWGN2}). Apparently, one can choose $l_u\equiv 0$ and $l_d\equiv 1$, yielding $\overline{I}_{\text{SK}}(\Theta;Y)=P/2$, which is widely known as the capacity of (\ref{AWGN2}).
}
\end{example}

\begin{example}\label{example-ihara}
{\rm When $l(t,s)= 1$ for any $t\ge s$, it follows from (\ref{resolvent-kernel}) that the corresponding $h(t,s)=-e^{s-t}$ for any $t\ge s$ and therefore the channel (\ref{ACGN-B-1}) boils down to
$$
Y(t)=\int_0^tX(s)ds+B(t)-\int_0^t\int_0^se^{u-s}dB(u)ds.
$$
Apparently, it can be verified that $l_u\equiv l_d\equiv c$, where $c$ is a non-zero constant. Thus, we have $\alpha=1,\beta=0$, yielding that $\bar{I}_{\text{SK}}(\Theta;Y)$ is the unique positive root of the cubic equation $P(x+1)^2=2x^3$. This recovers Proposition 1 in \cite{ihara1990capacity}.
}
\end{example}

To conclude this section, although Theorem \ref{lower-bound-CTGC} provides a lower bound on feedback capacity of a subclass of ACGN channels, this lower bound is somewhat implicit. In Section \ref{proof-main-thm}, we give a more explicit expression by narrowing our attention to the continuous-time ARMA$(1,1)$ Gaussian channel.

\section{Feedback Capacity}\label{proof-main-thm}
In this section, we focus on the following continuous-time ARMA$(1,1)$ Gaussian channel
\begin{equation}\label{ou-awgn-main}
Y(t)=\int_0^tX(s)ds+Z(t),
\end{equation}
where 
$$
Z(t)=B(t)+\lambda\int_0^t\int_{-\infty}^se^{-\kappa(s-u)}dB(u)ds,~\lambda\in\mathbb{R},~\kappa>0.
$$
The following theorem is our main result in which we derive an explicit formula for $C_{fb}(P)$.
\begin{thm}\label{main-thm}
$C_{fb}(P)$ is determined in the following two cases: 
\begin{itemize}
\item[(1)] if $\lambda\le -2\kappa$ or $\lambda\ge0$, then $C_{fb}(P)=P/2$ ;
\item[(2)] if $-2\kappa<\lambda<0$, then $C_{fb}(P)$ is the unique positive root of the following third-order polynomial equation:
\begin{equation}\label{3-poly}
P(x+\kappa)^2=2x(x+\vert\kappa+\lambda \vert)^2.
\end{equation}
\end{itemize} 
\end{thm}

The proof of the theorem is divided into two directions. Firstly, we demonstrate that any achievable rate $R$ must satisfy the condition $R \le C_{fb}(P)$, referred to as {\em the converse part}. Secondly, we establish the achievability of any rate $R$ with $R < C_{fb}(P)$, referred to as {\em the achievability part}.

Before the proof, we introduce two auxiliary random processes $Z_\dag=\{Z_\dag(t);t\in[0,T] \}$, $Z_\ast=\{Z_\ast(t);t\in[0,T] \}$, which are defined as
\begin{equation}\label{Z0-Zstar}
Z_\dag(t)=\int_0^te^{-\kappa(t-s)}dB(s), ~~~Z_\ast(t)=B(t)+\lambda\int_0^t\int_0^se^{-\kappa(s-u)}dB(u)ds,
\end{equation}
respectively. 
Note that $Z_\dag$ solves the stochastic differential equation 
$$
dX(t)=-\kappa X(t)dt+dB(t),~~X(0)=0.
$$
Thus, we obtain
\begin{equation}\label{Z0-form2}
Z_\dag(t)=B(t)-\kappa\int_0^tZ_\dag(s)ds.
\end{equation}
Moreover, it holds that
\begin{equation}\label{Zs-Z0-relation}
Z_\ast(t)=\left(1+\frac{\lambda}{\kappa}\right)B(t)-\frac{\lambda}{\kappa}Z_\dag(t).
\end{equation}

\subsection{Proof of the Converse Part}
In this subsection, we prove the converse part of Theorem \ref{main-thm}, using some existing results on the feedback capacity of discrete-time ARMA(1,1) Gaussian channels (see detailed definitions in \cite{cover1989gaussian}) under average power constraint $P$. 
 For such channels, Yang \emph{et al.} \cite[Theorem 7]{yang2007feedback} derived a relatively explicit formula for feedback capacity under the assumption that stationary inputs can achieve feedback capacity, which has been confirmed by Kim in the proof of \cite[Theorem 3.1]{kim2009feedback}.
 ~Thus, feedback capacity for discrete-time ARMA(1,1) Gaussian channels is known, as reformulated below.

\begin{thm}[\!\!\cite{yang2007feedback},\cite{kim2009feedback}\footnote{Theorem \ref{kim} has been stated and proved in \cite[Theorem 5.3]{kim2009feedback}. However, a recent paper \cite{derpich2022comments} pointed out that there is a gap in the proof of a key result \cite[Corollary 4.4]{kim2009feedback}, and as a consequence, the proof of Theorem 5.3 in \cite{kim2009feedback} is invalid. However, as we emphasized in the first paragraph of this subsection, the result in \cite[Theorem 5.3]{kim2009feedback} still holds.}]\label{kim}
 Suppose the noise process $\{Z_i\}$ is an ARMA(1,1) Gaussian process satisfying
\begin{equation*}\label{kim-arma}
Z_i +\phi Z_{i-1} = U_i +\theta U_{i-1}, ~~~ i\in \mathbb{Z}, ~~~ |\phi|<1
\end{equation*}
where $\{U_i \}$ is a white Gaussian process with zero mean and unit variance. Then, under the average power constraint
\[
\lim_{n\to\infty}\E\left[\frac{1}{n}\sum_{i=1}^n\vert X_i\vert^2\right]\le P,
\]
the feedback capacity of additive Gaussian channel $Y_i =X_i+Z_i,i=1,2,...$ is given by 
\begin{equation*}
C_{FB} = -\frac{1}{2}\log ({x^2_0}),
\end{equation*}
where $x_0$ is the unique positive root of the fourth-order polynomial equation:
\begin{equation}\label{thm-kim-P}
Px^2=\left\{\begin{aligned}
&\frac{(1-x^2)(1+\text{sgn}(\phi-\theta) \theta x)^2}{(1+\text{sgn}(\phi-\theta) \phi x)^2}&\text{if}~ \vert\theta \vert\le 1\\
&\frac{(1-x^2)(\theta+\text{sgn}(\phi-1/\theta)x)^2}{(1+\text{sgn}(\phi-1/\theta) \phi x)^2}&\text{if}~ \vert\theta \vert> 1.
\end{aligned}
\right.
\end{equation}
\end{thm}

\begin{remark}
Yang \emph{et al.} and Kim only gave the result for $\vert\theta\vert < 1$; the case $\vert \theta\vert >1$ can be readily proved by converting it into the case $\vert\theta\vert < 1$; the case $\vert \theta\vert =1$ can be easily established via a perturbation argument.
\end{remark}

Now, we can derive an upper bound on the $T$-block feedback capacity $C_{fb,T}(P)$ in the following lemma.

\begin{lem}\label{lemma:cft}
The $T$-block feedback capacity $C_{fb,T}(P)$ of the continuous-time ARMA(1,1) Gaussian channel (\ref{ou-awgn-main}) is upper bounded as
\begin{equation}\label{lem:cft}
C_{fb,T}(P)\le \left\{\begin{aligned}&\frac{P}{2}, &~~~\text{if}~\lambda \le -2\kappa~\text{or}~\lambda\ge 0; \\ x_0(&P;\lambda,\kappa),&~~~\text{if}~-2\kappa<\lambda<0  ,\end{aligned}\right.
\end{equation}
where $x_0(P;\lambda,\kappa)$ is the unique positive root of the polynomial equation (\ref{3-poly}).
\end{lem}

\begin{proof}
By Theorem \ref{detailed-structure}, we can prove (\ref{lem:cft}) by considering any Gaussian pair $(\Theta,X)$ of the form (\ref{opt-pair}) in which $\Theta=\{\Theta(t);t\in[0,T] \}$ is assumed to be Gaussian such that $\int_0^T\E[\Theta^2(t)]dt<\infty$, and $X$ satisfies the constraint (\ref{AP-T}) and takes the following form 
$$
X(t) = \Theta(t)-\E[\Theta(t)\vert \mathcal{F}_t(Y^\ast)],
$$
where 
\begin{equation}\label{Ystar}
Y^\ast(t)=\int_0^t\Theta(s)ds+Z(t),~~~t\in[0,T].
\end{equation}
\noindent Moreover, it is known \cite{ihara1980capacity} that there exists a Volterra kernel $K(t,s)$ on $L^2([0,T]^2)$ such that $\E[\Theta(t)\vert \mathcal{F}_t(Y^\ast)]=\int_0^t K(t,s)dY^\ast(s)$. The remainder of the proof is divided into three steps. In Steps 1 \& 2, we assume that the following condition: 
\begin{itemize}
\item[(C.0)] The Volterra kernel $K(t,s)$ is continuous on the set $\{(t,s)\in[0,T]^2;t\ge s \}$
\end{itemize}
is satisfied.

\textbf{Step 1.} In this step, we shall introduce a sequence of discrete-time ARMA$(1,1)$ Gaussian channels constructed from the continuous-time ARMA$(1,1)$ Gaussian channel (\ref{ou-awgn-main}) by using a discrete-time approximation method. 

For any $n\in\mathbb{N}$, we consider a partition $\{t_k^{(n)};k=0,1,...,n \}$ of $[0,T]$ satisfying $t^{(n)}_{k+1}-t^{(n)}_k=\delta_n$ for all $k$, where $\delta_n=T/n$. Define $\{B^{(n)}_k;k=0,1,...,n-1\}, \{Z_{k}^{(n)};k=0,1,...,n-1 \}$ as
\begin{align}
B^{(n)}_k&=B(t^{(n)}_{k+1})-B(t^{(n)}_k),\nonumber\\
Z_{k}^{(n)}&=B^{(n)}_k+\lambda d^{(n)}_k\left(\zeta_0+\sum_{i=0}^{k-1}e^{\kappa t^{(n)}_{i+1}}B^{(n)}_i\right)\label{Z-app-1}
\end{align}
respectively, where $d_k^{(n)}\triangleq \int_{t^{(n)}_k}^{t^{(n)}_{k+1}} e^{-\kappa s}ds$ and $\zeta_0=\int_{-\infty}^0 e^{\kappa s}dB(s)$. Then, it can be verified that $\{Z^{(n)}_k/\sqrt{\delta_n};k=0,1,...,n-1 \}$ is an ARMA(1,1) Gaussian process satisfying
\begin{equation*}\label{arma11-non-stationary}
\frac{Z^{(n)}_{k+1}}{\sqrt{\delta_n}}=e^{-\kappa \delta_n}\frac{Z^{(n)}_{k}}{\sqrt{\delta_n}}+\frac{B^{(n)}_{k+1}}{\sqrt{\delta_n}}+\left(\frac{\lambda}{\kappa}-\left(\frac{\lambda}{\kappa}+1\right)e^{-\kappa \delta_n}\right)\frac{B^{(n)}_k}{\sqrt{\delta_n}},~\text{for}~ k=0,1,...,n-1,
\end{equation*}
which, however, is not stationary. It turns out that we can modify (\ref{Z-app-1}) to guarantee stationarity. Specifically, we define $\{\widetilde{Z}^{(n)}_k;k=0,1,...,n-1 \}$ as
\begin{equation}\label{z-app-stat}
\widetilde{Z}^{(n)}_k\triangleq B^{(n)}_k+\lambda d^{(n)}_k \left(m(\delta_n)\zeta_0+\sum_{i=0}^{k-1}e^{\kappa t^{(n)}_{i+1}}B^{(n)}_i\right),
\end{equation}
where $m(x)\triangleq \sqrt{2\kappa x/(1-e^{-2\kappa x})}$. It is straightforward to verify that $\{\widetilde{Z}_k^{(n)}/\sqrt{\delta_n};k=0,..,n-1\}$ is a stationary ARMA(1,1) process of the following form
\begin{equation}\label{arma11-stationary}
\frac{\widetilde{Z}^{(n)}_{k+1}}{\sqrt{\delta_n}}=e^{-\kappa\delta_n}\frac{\widetilde{Z}^{(n)}_{k}}{\sqrt{\delta_n}}+\frac{B^{(n)}_{k+1}}{\sqrt{\delta_n}}+\left(\frac{\lambda}{\kappa}-\left(\frac{\lambda}{\kappa}+1\right)e^{-\kappa \delta_n} \right)\frac{B^{(n)}_k}{\sqrt{\delta_n}}.
\end{equation}
Furthermore, we define $\{Y^{\ast,(n)}_{k} \}$ and $\{Y^{(n)}_{k} \}$ as
\begin{align}
Y^{\ast,(n)}_{k}&=\Theta^{(n)}_k+\widetilde{Z}^{(n)}_{k},\label{discrete-GC-nofeedback}\\
Y^{(n)}_{k}&=\Theta^{(n)}_k-\zeta^{(n)}_k+\widetilde{Z}^{(n)}_{k},\label{discrete-GC1-feedback}
\end{align}
where $\{\Theta^{(n)}_k\}$ and $\{\zeta^{(n)}_k \}$ are defined by
\begin{align*}
\Theta^{(n)}_k=\int_{t^{(n)}_k}^{t^{(n)}_{k+1}}\Theta(s)ds,\quad \zeta^{(n)}_k=\delta_n \sum_{i=0}^{k-1}K(t^{(n)}_k,t^{(n)}_i)Y^{\ast,(n)}_{i},
\end{align*}
respectively. Note that (\ref{discrete-GC1-feedback}) and (\ref{discrete-GC-nofeedback}) correspond to $n$-block discrete-time ARMA(1,1) Gaussian channels with feedback and without feedback, respectively. 

\textbf{Step 2.} 
This step will be devoted to approximating $P/2$ and $x_0(P;\lambda,\kappa)$ by feedback capacities of the sequence of ARMA(1,1) Gaussian channels (\ref{discrete-GC1-feedback}).

To this end, we have the following chain of inequalities:
\begin{equation}
\begin{aligned}\label{proof-converse-part-Z1}
\frac{1}{T}I(\Theta_0^T;Y_0^T)&\overset{(a)}{=}\frac{1}{T}I(\Theta_0^T;Y_0^{\ast,T})\\
&\overset{(b)}{\le}\varliminf_{n\to\infty}\frac{1}{T}I(\{\Theta_k^{(n)}\};\{Y_k^{\ast,(n)}\})\\
&\overset{(c)}{=}\varliminf_{n\to\infty} \frac{1}{T}I(\{\Theta_k^{(n)}\};\{Y_k^{(n)}\})\\
&\overset{(d)}{\le}\varliminf_{n\to\infty}\frac{1}{\delta_n}C_{FB,n}(P\delta_n+\frac{e(\delta_n)}{n})\\
&\overset{(e)}{\le} \varliminf_{n\to\infty}\frac{1}{\delta_n}C_{FB}(P\delta_n+\frac{e(\delta_n)}{n}) \\
&\overset{(f)}{\le}\left\{\begin{aligned}&\frac{P}{2} &~~~\text{if}~\lambda \le -2\kappa~\text{or}~\lambda\ge 0 \\ x_0(&P;\lambda,\kappa)&~~~\text{if}~-2\kappa<\lambda<0  ,\end{aligned}\right.
\end{aligned}
\end{equation}
where $e(\delta_n)$ is some function (to be specified later) dependent on $\{\Theta(t)\}$ with the property $\lim_{n\to\infty}e(\delta_n)=0$, where $C_{FB,n}(P)$ denotes the $n$-block feedback capacity \cite{cover1989gaussian} of the channel (\ref{discrete-GC1-feedback}) under the constraint that the average power of the channel input is bounded by $P$ (see \cite{kim2009feedback}) and $C_{FB}(P)$ denotes feedback capacity. Now, with $(a)$-$(f)$ validified (proofs can be founded in Appendix \ref{proof-of-cf}), (\ref{lem:cft}) immediately follows from (\ref{proof-converse-part-Z1}) and Theorem \ref{detailed-structure}.

\textbf{Step 3.} We will prove in this step that the continuity assumption (C.0) can be dropped. Indeed, there exists a sequence of Volterra kernels $\{K_{(m)};m=1,2,...\}$ satisfying (C.0) and 
$$
\lim_{m\to\infty}\| K_{(m)}-K\|_{2}=0.
$$
Set
$$
P_m\triangleq\frac{1}{T}\int_0^T\E\left[\left\vert\Theta(t)-\int_0^tK_{(m)}(t,s)dY^\ast(s)\right\vert^2\right]dt.
$$
Then, we have
\begin{equation}\label{P_m}
\lim_{m\to\infty}P_m=\frac{1}{T}\int_0^T\E\left[\left\vert\Theta(t)-\int_0^tK(t,s)dY^\ast(s)\right\vert^2\right]dt\le P,
\end{equation}
where we have used the fact
$$
\lim_{m\to\infty}\int_0^T\E\left[\left\vert\int_0^t K_{(m)}(t,s)dY^\ast(s)-\int_0^tK(t,s)dY^\ast(s)\right\vert^2\right]dt=0.
$$
Note that $(c)$-$(f)$ in Step 2 hold true for any continuous Volterra kernel function $K$. Thus, replacing $K$ in (\ref{proof-converse-part-Z1}) by $K^{(m)}$ in the derivation of $(c)$-$(f)$, we obtain
$$
\frac{1}{T}I(\Theta_0^T;Y_0^T)\le\left\{\begin{aligned}&\frac{P_m}{2} &~~~\text{if}~\lambda \le -2\kappa~\text{or}~\lambda\ge 0 \\ x_0(&P_m;\lambda,\kappa)&~~~\text{if}~-2\kappa<\lambda<0  ,\end{aligned}\right.,
$$
which, together with (\ref{P_m}), establishes the same inequality (\ref{proof-converse-part-Z1}).
\end{proof}

The following corollary is an immediate consequence of Lemma \ref{lemma:cft}.
\begin{co}\label{cor-converse}
It holds that
\begin{equation}\label{max-mutu-rate}
C_{fb,\infty}(P)\le \left\{\begin{aligned}\frac{P}{2} ,~~~\text{if}~\lambda \le -2\kappa~\text{or}~\lambda\ge 0; \\ x_0(P;\lambda,\kappa),~~~\text{if}~-2\kappa<\lambda<0 ,\end{aligned}\right.
\end{equation}
\end{co}
\begin{proof}
For any input $(\Theta,X)$ satisfying (\ref{AP-infinite}), there exists a function $e_P(T)$ with $\lim_{T\to\infty}e_P(T)=0$ such that
$$
\frac{1}{T}\int_0^T\E[X^2(t)]dt\le P+e_P(T)
$$
for all $T$. By the definition of $C_{fb,T}(P)$, we obtain
\[
\overline{I}(\Theta;Y)\le \varlimsup_{T\to\infty}C_{fb,T}(P+e_P(T)).
\]
Thus, (\ref{max-mutu-rate}) immediately follows from (\ref{lem:cft}) and the continuity of $P/2$ and $x_0(P;\lambda,\kappa)$ in $P$. 
\end{proof}
\begin{proof}[Proof of the Converse Part]
Clearly, it follows from Lemma \ref{lemma:cft} that 
$$
\lim_{T\to\infty}\frac{1}{T}C_{fb,T}(P) = 0,
$$
which, together with Theorem \ref{coding-theorem}, immediately implies that any achievable rate $R$ must satisfy the condition $R \le C_{fb,\infty}(P)$. This, combined with Corollary \ref{cor-converse}, establishes the proof of the converse part.
\end{proof}

\subsection{Proof of the Achievability Part}\label{sec:proof-of-lower-bound}
We first use Theorem \ref{lower-bound-CTGC} to prove Case (2) in Theorem \ref{main-thm}, as detailed below.
\begin{proof}[Case (2)]
Note that $Z$ can be regarded as the solution of the following stochastic differential equation:
\begin{equation*}\label{Z-sde}
dZ(t)=\lambda(Z_\dag(t)+\zeta_0 e^{-\kappa t})dt+dB(t),~~~ Z(0)=0.
\end{equation*}
Set $\eta(t)=\lambda(Z_\dag(t)+\zeta_0 e^{-\kappa t})$. Then, the covariance function of $\{\eta(t)\}$ can be computed as
$$
\gamma_{\eta}(s,t)=\frac{\lambda^2}{2\kappa}e^{-\kappa\vert s-t \vert},
$$
which is continuous at $s=t$. By \cite[Theorem 7.15]{lipster1977statistics}, it holds that $\mu_{Z_0^T}\sim\mu_{B_0^T}$ for all $T$. It then follows from (\ref{Gaussian-e-Brownian})-(\ref{l-kernel}) that there exists a standard Brownian motion $\{V(t)\}$ on $(\Omega,\{\mathcal{F}_t(Z) \},\mathbb{P})$ such that
\begin{equation}\label{goal-Z-h}
V(t)=Z(t)+\int_0^t\int_0^sl_z(s,u)dZ(u)ds,
\end{equation} 
where the Volterra kernel function $l_z(s,u)\in L^2([0,T]^2)$ for all $T$.

We now evaluate the resolvent kernel $l_z(s,u)$ and prove that $l_z$ fulfills all the conditions in Theorem \ref{lower-bound-CTGC}.
It follows from (\ref{Z0-Zstar}) and (\ref{Z0-form2}) that 
\begin{equation}\label{Z-star-B}
Z_\ast(t)=\int_0^tF_\ast(t,u)dB(u),
\end{equation}
where $F_\ast(t,u)$ is a Volterra kernel function satisfying $F_\ast(t,u)=1+\int_u^t f_\ast(s,u)ds$ for any $t\ge u$ and the Volterra kernel function $f_\ast(s,u)=\lambda e^{-\kappa(s-u)}$ for any $s\ge u$. 
Since the resolvent kernel $ g_\ast(s,u)$ of $f_\ast(s,u)$ is computed as
\begin{align*}
g_\ast(s,u)= -\lambda e^{(\kappa+\lambda)(u-s)} \quad \mbox{for all $s\geq u$},
\end{align*}
then, by (\ref{Gaussian-e-Brownian})-(\ref{l-kernel}), we obtain
\begin{align*}
B(t)=\int_0^t G_\ast(t,u)dZ_\ast(u),\label{76}
\end{align*}
where $G_\ast(t,u)$ is the Volterra kernel function satisfying $G_\ast(t,u)=1+\int_u^tg_\ast(s,u)ds$ for any $t\ge u$. Therefore, by (\ref{Z-star-B}), we have
\begin{align*}
Z(t)&=Z_\ast(t)+\lambda\int_0^te^{-\kappa s}ds \zeta_0=\int_0^tF_\ast(t,u)d\Phi(u),
\end{align*}
where 
\begin{equation}\label{Phi}
\Phi(t)=\int_0^t\lambda \zeta_0e^{-(\kappa+\lambda)s}ds+B(t).
\end{equation}
Then, since $g_\ast(s,u)$ is the resolvent kernel of $f_\ast(s,u)$, it holds that
\begin{equation}\label{Phi-Z}
\Phi(t)=\int_0^tG_\ast(t,u)dZ(u).
\end{equation}
By \cite[Lemma 6.2.6]{oksendal2013stochastic}, the innovation process $\{V(t) \}$ defined by
\begin{equation}\label{V-Brownian}
V(t)=\Phi(t)-\lambda\int_0^te^{-(\kappa+\lambda)s}\E[\zeta_0\vert \mathcal{F}_s(\Phi)]ds
\end{equation}
is a standard Brownian motion. 
The one-dimensional Kalman-Bucy filter \cite[Theorem 6.2.8]{oksendal2013stochastic} is applied to estimate $\zeta_0$ from the observation equations (\ref{Phi}) to yield the following estimate of $\zeta_0$:
\begin{equation}\label{est-zeta}
\E[\zeta_0\vert \mathcal{F}_t(\Phi)]=\frac{\int_0^t\lambda e^{-(\kappa+\lambda)s}d\Phi(s)}{2\kappa+\int_0^t\lambda^2 e^{-2(\kappa+\lambda)s}ds}.
\end{equation}
Substituting (\ref{Phi-Z}) and (\ref{est-zeta}) into (\ref{V-Brownian}) and going through a series of elementary calculations, we obtain (\ref{goal-Z-h}), where $l_z(s,u)$ is calculated by
\begin{equation*}
l_z(s,u)=\left\{
\begin{aligned}
&\frac{\lambda(2\kappa+\lambda)^2e^{(\kappa+\lambda)u}+\lambda^2(2\kappa+\lambda)e^{-(\kappa+\lambda)u}}{\lambda^2e^{-(\kappa+\lambda)s}-(\lambda+2\kappa)^2e^{(\kappa+\lambda)s}},~~&\text{if}~\lambda+\kappa\neq 0;\\
&\frac{\kappa(\kappa u+1)}{\kappa s+2},~~ &\text{if}~ \lambda+\kappa=0.
\end{aligned}
\right.
\end{equation*}
It is easy to see that $l_z$ satisfies all the conditions in Theorem \ref{lower-bound-CTGC}. Then, the corresponding $\alpha_z,\beta_z$ are given by
$$
\alpha_z=\left\{
\begin{aligned}
&\kappa,~~&\text{if}~\kappa+\lambda=0;\\
&\lambda+2\kappa,~~&\text{if}~\kappa+\lambda<0;\\
&-\lambda,~~&\text{if}~\kappa+\lambda>0,
\end{aligned}
\right.
$$
and
$$
\beta_z=\left\{
\begin{aligned}
&0,~~&\text{if}~\kappa+\lambda=0;\\
&-(\lambda+\kappa),~~&\text{if}~\kappa+\lambda<0;\\
&\lambda+\kappa,~~&\text{if}~\kappa+\lambda>0,
\end{aligned}
\right.
$$
respectively. By Theorem \ref{lower-bound-CTGC}, we have $\overline{I}_\text{SK}(\Theta;Y)=Pr_z^2$, where $r_z$ is one of the real roots of the following cubic equation:
\begin{equation}\label{OU-3py-1}
-Py^3+\frac{P}{\sqrt{2}}y^2-\vert\lambda+\kappa \vert y+\frac{1}{\sqrt{2}} \kappa=0.
\end{equation}
It is not difficult to see that the equation (\ref{OU-3py-1}) has the unique positive root for all $-2\kappa\le\lambda\le0$. Then, substituting $y=\sqrt{x/P}$ into (\ref{OU-3py-1}), we are able to prove that $\overline{I}_\text{SK}(\Theta;Y)$ is the unique positive root of the third-order polynomial (\ref{3-poly}), which implies 
\[
C_{fb,\infty}(P)\geq x_0(P;\lambda,\kappa).
\]
This, together with Corollary \ref{cor-converse} and Theorem \ref{coding-theorem}, immediately yields 
\begin{equation}\label{achivev-case-2}
C_{fb}(P)\ge x_0(P;\lambda,\kappa),
\end{equation}
as desired.
\end{proof}

\begin{remark}
Here we rewrite $C_{fb}(P)$ as $C_{fb}(P;\lambda,\kappa)$ to emphasize its dependence on $\lambda,\kappa$. It follows from the definition of $x_0(P;\lambda,\kappa)$ that $ x_0(P;\lambda,\kappa)>P/2$ for $-2\kappa<\lambda<0, \kappa>0$, which, together with (\ref{achivev-case-2}), immediately implies that $C_{fb}(P;\lambda,\kappa)>C_{fb}(P;0,\kappa)$,
where $C_{fb}(P;0,\kappa)$ is feedback capacity of an AWGN channel (\ref{AWGN2}).
In other words, for a continuous-time ARMA$(1,1)$ Gaussian channel, ``coloring'' may increase capacity.
\end{remark}

Next, to establish the achievability of Case (1), we will present a coding theorem for $C_{nfb}(P)$. To achieve this, as before, instead of transmitting a message index $W$, a random variable taking values from a finite alphabet, we will transmit a message process $ \Theta = \{ \Theta(t) \}$, a real-valued random process. Then, compared to (\ref{coding-thm-input}), the channel input $\{X(t)\}$
will take the form: 
\begin{equation}\label{coding-thm-input-2}
X(t)=g_t(\Theta(t)),
\end{equation}
which is independent of $\{ Z(t)\}$. Now, let us turn our attention to the continuous-time ARMA$(1,1)$ Gaussian channel (\ref{OU-WGN}) reformulated as
\begin{equation}\label{carma11-channel-g}
\boldsymbol{y}(t) = \boldsymbol{x}(t)+\boldsymbol{z}(t), \quad \mbox{$-\infty<t<+\infty$},
\end{equation}
where 
$$
\boldsymbol{z}(t)=\boldsymbol{w}(t)+\lambda \boldsymbol{u}(t),
$$
where $\boldsymbol{w}(t)=\dot{B}(t)$ and $\boldsymbol{u}(t)=\int_{-\infty}^te^{-\kappa(t-u)}dB(u)$. Here we remark that, in the most rigorous terms, the channel (\ref{carma11-channel-g}) should be interpreted as
$$
\boldsymbol{y}(\phi)=\boldsymbol{x}(\phi)+\boldsymbol{z}(\phi),\quad \phi\in\mathcal{D}.
$$
Now, for any $T$, let $\mathcal{D}_T=\{\phi\in\mathcal{D}:\text{supp}(\phi)\subset [0,T]  \}$, and define
$$
I(\boldsymbol{x}_0^T;\boldsymbol{y}_0^T)=\sup I(\boldsymbol{x}(\phi_1),...,\boldsymbol{x}(\phi_m);\boldsymbol{y}(\varphi_1),...,\boldsymbol{y}(\varphi_n)),
$$
where the supremum is taken over all positive integers $m,n$ and all test functions $\{\phi_1^m,\varphi_1^n\}\subset \mathcal{D}_T$. Then, we consider the so-called {\em $T$-block non-feedback capacity}
$$
C_{nfb,T}(P)=\sup \frac{1}{T}I(\boldsymbol{x}_0^T;\boldsymbol{y}_0^T),
$$
where the supremum is taken over all generalized random processes $\{\boldsymbol{x}(\phi);\phi\in\mathcal{D}\}$ such that
$$
\boldsymbol{x}(\phi)=\int_\mathbb{R}X(t)\phi(t)dt,\quad \phi\in\mathcal{D}_T,
$$
where $\{X(t)\}$ is given by (\ref{coding-thm-input-2}) and subject to the average power constraint
$$
\frac{1}{T}\int_0^T\E[X^2(t)]dt \le P.
$$
Furthermore, we define
$$
\overline{I}(\boldsymbol{x};\boldsymbol{y})=\limsup_{T\to\infty}\frac{1}{T}I(\boldsymbol{x}_0^T;\boldsymbol{y}_0^T),
$$
provided the limit exists, and further define 
$$
C_{nfb,\infty}(P)=\sup \overline{I}(\boldsymbol{x};\boldsymbol{y}),
$$
where the supremum is taken over all generalized random processes $\{\boldsymbol{x}(\phi);\phi\in\mathcal{D}\}$ such that
$$
\boldsymbol{x}(\phi)=\int_\mathbb{R}X(t)\phi(t)dt,\quad \phi\in\mathcal{D}_\infty,
$$
where $\{X(t)\}$ is given by (\ref{coding-thm-input-2}) and subject to the constraint
$$
\varlimsup_{T\to\infty}\frac{1}{T}\int_0^T\E[X^2(t)]dt\le P.
$$
Now, we present the aforementioned coding theorem for $C_{nfb}(P)$ below.
\begin{thm}[{\cite[Theorem 1]{Ihara1992channelcodingthm}}]\label{coding-theorem-2}
Assume that
\begin{equation}\label{cond:nfb}
\lim_{T\to\infty}\frac{1}{T}C_{nfb,T}(P)=0.
\end{equation}
If $R< C_{nfb,\infty}(P)$ and $C_{nfb,\infty}(P)$ is continuous in $P$, then the rate $R$ is achievable. Conversely, if a rate $R$ is achievable, then $R\le C_{nfb,\infty}(P)$.
\end{thm}

Then, the proof of achievability of Case (1) will use the following proposition, which gives an explicit formula for $C_{nfb}(P)$.
\begin{pr}\label{prop-lb-2p}
It holds that
\begin{equation}\label{lem:C0-CASE1}
C_{nfb}(P)= \frac{P}{2},
\end{equation}
for all $\lambda\ge0$ or $\lambda\le-2\kappa,$ $\kappa>0$.
\end{pr}

The proof of Proposition \ref{prop-lb-2p} relies on the following result, whose proof is very similar to that of \cite[Lemma 4]{ihara1990capacity} and thus omitted.
\begin{lem}\label{sdf-ou-colored}
$\{\boldsymbol{z}(t) \}$ is a generalized stationary Gaussian process with spectral density function 
\begin{equation*}\label{sdf-ou-coloered-fun}
S_{\boldsymbol{z}}(x)=\frac{4\pi^2 x^2+(\kappa+\lambda)^2}{4\pi^2 x^2+\kappa^2},~x\in\mathbb{R},
\end{equation*}
where $\lambda\in\mathbb{R},~\kappa>0$.
\end{lem}

\begin{proof}[Proof of Proposition \ref{prop-lb-2p}]
We see from Lemma \ref{sdf-ou-colored} that $S_{\boldsymbol{z}}(x)$ is rational. Thus, applying the formula for non-feedback capacity of the continuous-time stationary ACGN channel (see Theorem 5 in \cite{baker1991information}) to the channel (\ref{carma11-channel-g}), we have
$$
C_{nfb,\infty}(P)=\frac{P}{2}.
$$
which, together with Theorem \ref{coding-theorem-2}, immediately implies (\ref{lem:C0-CASE1}), as desired.
\end{proof}

Now we can prove the achievability of Case (1).
\begin{proof}[Case (1)]
The achievability part follows immediately from $C_{nfb}(P)\le C_{fb}(P)$ together with Proposition \ref{prop-lb-2p}.
\end{proof}

\section{Effect of Feedback}\label{sec:C0-Cf}
For a discrete-time ACGN channel, Cover and Pombra \cite{cover1989gaussian} have proved that feedback does not increase capacity ``much''; more precisely, feedback can increase capacity by at most half a bit, which is often referred to as {\em the half-bit bound} in the literature. It has been established in the same paper that feedback can at most double capacity, which was attributed to Pinsker \cite{pinsker1969talk} and Ebert \cite{ebert1970capacity} and often referred to as {\em the factor-of-two bound} in the literature. In the discrete-time setting, there are many extensions and improvements to the half-bit and factor-of-two bounds (see, e.g., \cite{chen1999refinements}). 

By comparison, for the continuous-time ACGN channel (\ref{ACGN-B-1}), Ihara showed in \cite{ihara1991mutual} that the factor-of-two bound also holds and further pointed out that it is tight by investigating a special case of ACGN channels as in Example \ref{example-ihara}. This section will examine the effect of feedback on the capacity of the continuous-time ARMA$(1,1)$ Gaussian channel.

Via a nuanced analysis, we characterize the factor-of-two bound for the continuous-time ARMA$(1,1)$ Gaussian channel (\ref{OU-WGN}); in particular, our result reveals when the feedback capacity equals or doubles the non-feedback capacity for a continuous-time ARMA$(1,1)$ Gaussian channel.

\begin{thm}\label{thm:C0-Cf}
For the continuous-time ARMA$\it (1,1)$ Gaussian channel (\ref{OU-WGN}), it holds that:
\begin{itemize}
\item[(1)] if $\lambda \le -2\kappa$ or $\lambda\ge 0$, then $C_{nfb}(P)=C_{fb}(P)$;
\item[(2)] if $-2\kappa<\lambda<0$, then
\[
C_{nfb}(P)<C_{fb}(P)\le 2C_{nfb}(P),
\]
where the equality holds if and only if $\lambda=-\kappa$ and $P=\kappa/2$.
\end{itemize}
\end{thm}

\begin{proof}
Throughout this proof, for clarity of presentation, we will rewrite $C_{nfb}(P)$ and $C_{fb}(P)$ as $C_{nfb}(P;\lambda,\kappa)$ and $C_{fb}(P;\lambda,\kappa)$, respectively.

Obviously, (1) immediately follows from Theorem \ref{main-thm} and Proposition \ref{prop-lb-2p}. Thus, it remains to prove (2). 

For any $-2\kappa<\lambda<0$ and $\kappa>0$, noting that $S_{\boldsymbol{z}}(x;\lambda,\kappa)=S_{\boldsymbol{z}}(x;-2\kappa-\lambda,\kappa)$,  we deduce from \cite[Theorem 5]{baker1991information} that for any $P$,
$$
C_{nfb}(P;\lambda,\kappa) = C_{nfb}(P;-2\kappa-\lambda,\kappa).
$$
Furthermore, by Theorem \ref{main-thm}, we also have 
$$
C_{fb}(P;\lambda,\kappa) = C_{fb}(P;-2\kappa-\lambda,\kappa).
$$
Hence, it suffices to prove (2) under the condition $-\kappa\le\lambda<0$ and $\kappa>0$. 

To this end, we consider the following two cases:\\
\textbf{Case 1:} $\lambda=-\kappa,~\kappa>0$. We first show that $C_{fb}(\kappa/2;-\kappa,\kappa)=2C_{nfb}(\kappa/2;-\kappa,\kappa)$. Note that, by \cite[Theorem 5]{baker1991information} and Theorem \ref{coding-theorem-2}, we have
\begin{equation}\label{proof-cocf-1}
C_{nfb}(P;-\kappa,\kappa)=\left\{\begin{aligned}
&y_P, &\qquad \text{if}~~ P \le \frac{k}{2};\\
&\frac{P}{2}+\frac{\kappa}{4}, &\text{if}~~ P\ge \frac{k}{2},
\end{aligned}
\right.
\end{equation}
where $y_P$ is uniquely determined by 
\begin{equation}\label{proof-c0cf-2}
\sin\left(\frac{2\pi}{\kappa}y_P\right) = \frac{2\pi}{\kappa}\left(y_P - P\right). 
\end{equation}
Recalling from Theorem \ref{main-thm} that $C_{fb}(P;-\kappa,\kappa)$ is the unique positive root of the following cubic equation
\begin{equation}\label{proof-c0cf-cf1}
P(x+\kappa)^2 = 2x^3,
\end{equation}
we deduce from (\ref{proof-cocf-1}) that 
\begin{equation}\label{proof-c0cf-k/2}
C_{fb}(\kappa/2;-\kappa,\kappa) =2C_{nfb}(\kappa/2;-\kappa,\kappa)=\kappa,
\end{equation}
as desired. 

In the following, we only prove $C_{nfb}(P;-\kappa,\kappa)<C_{fb}(P;-\kappa,\kappa)$ for all $P$, since $C_{fb}(P;-\kappa,\kappa)< 2C_{nfb}(P;-\kappa,\kappa)$ for $P\neq \kappa/2$ can be similarly handled. 

Now, 
 letting
\begin{equation}\label{Gamma-1}
\Gamma(x;\kappa) =\left\{\begin{aligned}
&\sin\left(\frac{2\pi}{\kappa}x\right)-\frac{2\pi}{\kappa}\left[x -\frac{2x^3}{(x+\kappa)^2}\right], & \quad 0< x < \frac{\kappa}{2};\\
&2\left(\frac{x}{2}+\frac{\kappa}{4}\right)^3 - x\left(\frac{x}{2}+\frac{5\kappa}{4}\right)^2,& \quad x\ge\frac{k}{2},
\end{aligned}
\right.
\end{equation}
we claim that for any $x>0$
\begin{equation}\label{proof-gamma-less-than-0}
\Gamma(x;\kappa)<0. 
\end{equation}
To see this, first note that it follows from (\ref{Gamma-1}) that for any $x> \kappa/2$
$$
\frac{d\Gamma(x;\kappa)}{dx}=-\frac{\kappa(14x+11\kappa)}{8}<0
$$
which, combined with the fact that $\Gamma(\kappa/2;\kappa)=-7\kappa^3/8<0$, immediately implies that $\Gamma(x;\kappa)<0$ for all $x\ge \kappa/2$. Furthermore, it can be readily verified that $d\Gamma(x;1)/dx <0$ for all $0<x<1/2$. This, together with the facts that $\Gamma(0;1)=0$ and $\Gamma(x/\kappa;1)=\Gamma(x;\kappa)$ for all $0<x<\kappa/2$, immediately implies that $\Gamma(x;\kappa)<0$ for all $0<x<\kappa/2$, as desired.

Next, we shall prove that (\ref{proof-gamma-less-than-0}) implies that $C_{nfb}(P;-\kappa,\kappa)<C_{fb}(P;-\kappa,\kappa)$ for all $P$. Clearly, it follows from (\ref{proof-cocf-1}) and (\ref{proof-c0cf-cf1}) that for any $P\ge\kappa/2$,
$$
\Gamma(P;\kappa)<0 \quad\Longleftrightarrow \quad C_{nfb}(P;-\kappa,\kappa)<C_{fb}(P;-\kappa,\kappa).
$$
Moreover, we claim that for any $0<y_P,P_y<\kappa/2$
\begin{equation}\label{proof-c0cf-Gamma}
\begin{aligned}
\Gamma(y_P;\kappa)&=0 \quad\Longleftrightarrow \quad C_{nfb}(P_y;-\kappa,\kappa)=C_{fb}(P_y;-\kappa,\kappa);\\
\Gamma(y_P;\kappa)&<0 \quad \Longleftrightarrow \quad C_{nfb}(P_y;-\kappa,\kappa)<C_{fb}(P_y;-\kappa,\kappa),
\end{aligned}
\end{equation}
where $y_P, P_y$ are uniquely determined from each other by 
$$
\sin\left(\frac{2\pi}{\kappa}y_P\right) = \frac{2\pi}{\kappa}\left(y_P - P_y\right).
$$
To see this, first note that when $\Gamma(y_P;\kappa)<0$ for some $0<y_P<\kappa/2$, by the definition of $\Gamma(y_P;\kappa)$, $P_y$ must satisfy 
$$
P_y>\frac{2y_P^3}{(y_P+\kappa)^2},
$$
which, together with (\ref{proof-cocf-1})-(\ref{proof-c0cf-cf1}), immediately implies that $C_{nfb}(P_y;-\kappa,\kappa)<C_{fb}(P_y;-\kappa,\kappa)$. Conversely, assuming that $C_{nfb}(P_y;-\kappa,\kappa)<C_{fb}(P_y;-\kappa,\kappa)$ for some $0<P_y<\kappa/2$, we have 
\begin{align*}
\Gamma(y_P;\kappa) &= \sin\left(\frac{2\pi}{\kappa} y_P\right)-\frac{2\pi}{\kappa} \left[y_P -\frac{2y_P^3}{(y_P+\kappa)^2} \right]\\
&= \frac{2\pi}{\kappa}\left[\frac{2y_P^3}{(y_P+\kappa)^2}-P_y \right]\\
&<0,
\end{align*}
where the second equality follows from the fact that 
$$
y_P= C_{nfb}(P_y;-\kappa,\kappa)<C_{fb}(P_y;-\kappa,\kappa),
$$
as desired. 

Therefore, we have shown that $C_{nfb}(P)<C_{fb}(P)$ for all $P$ if and only if $\Gamma(x;\kappa)<0$ for all $x>0$. This, together with (\ref{proof-gamma-less-than-0}), immediately completes the proof in this case.

\noindent{\textbf{Case 2:}}
Now, we turn to the case $-\kappa<\lambda<0$ and $\kappa>0$. Let
$$
P_{\max} = \frac{-\lambda(\lambda+2\kappa)}{2\kappa}.
$$
Then, by \cite[Theorem 5]{baker1991information} and Theorem \ref{coding-theorem-2}, we have
\begin{equation}\label{proof-c0cf-c0-formula}
C_{nfb}(P;\lambda,\kappa)=\left\{\begin{aligned}
&\frac{k}{\pi}\arctan\left(\frac{x_P}{\kappa}\right) -\frac{\kappa+\lambda}{\pi}\arctan\left(\frac{x_P}{\kappa+\lambda}\right),&\text{if}~~P<P_{\max};\\
&-\frac{\lambda}{2}+\frac{1}{2}(P-P_{\max}),&\text{if}~~ P\ge P_{\max},
\end{aligned}
\right.
\end{equation}
where $x_P$ is uniquely determined by 
\begin{equation}\label{proof-c0cf-xp}
\frac{\lambda(2\kappa+\lambda)}{\pi}\frac{x_P}{x_P^2+\kappa^2} - \frac{\lambda(2\kappa+\lambda)}{\pi\kappa}\arctan\left(\frac{x_P}{\kappa}\right) = P.
\end{equation}
On the other hand, recall from Theorem \ref{main-thm} that $C_{fb}(P;\lambda,\kappa)$ is the positive root of the following cubic equation
$$
P(x+\kappa)^2 = 2x(x+\kappa+\lambda)^2.
$$

Similarly as in Case 1, we only prove that for all $P$,
\begin{equation*}
C_{nfb}(P;\lambda,\kappa)<C_{fb}(P;\lambda,\kappa),
\end{equation*}
since the inequality $C_{fb}(P;\lambda,\kappa)<2C_{nfb}(P;\lambda,\kappa)$ can be proved in a parallel manner. 

To this end, we define two functions $\Lambda_1(P;\lambda,\kappa)$ over $[P_{\max},\infty)$ and $\Lambda_2(x;\lambda,\kappa)$ over $(0,\infty)$ as
\begin{equation*}
\begin{aligned}
\Lambda_1(P;\lambda,\kappa) &= 2C_{nfb}(P)(C_{nfb}(P)+\kappa+\lambda)^2 - P(C_{nfb}(P)+\kappa)^2 ;\\
\Lambda_2(x;\lambda,\kappa) &= \frac{\lambda(2\kappa+\lambda)}{\pi}\frac{x}{x^2+\kappa^2} - \frac{\lambda(2\kappa+\lambda)}{\pi\kappa}\arctan\left(\frac{x}{\kappa} \right) -\frac{2A(A+\kappa+\lambda)^2}{(A+\kappa)^2},
\end{aligned}
\end{equation*}
where 
$$
A = \frac{\kappa}{\pi}\arctan\left(\frac{x}{\kappa}\right) - \frac{\kappa+\lambda}{\pi}\arctan\left(\frac{x}{\kappa+\lambda}  \right).
$$
Now, using a largely parallel argument as in the proof of (\ref{proof-c0cf-Gamma}), we conclude that for any $P\ge P_{\max}$
$$
\Lambda_1(P;\lambda,\kappa) < 0 \quad\Longleftrightarrow \quad C_{nfb}(P;\lambda,\kappa)<C_{fb}(P;\lambda,\kappa),
$$
and moreover, for any $P_x<P_{\max}$,
$$
\Lambda_2(x_P;\lambda,\kappa)>0 \quad\Longleftrightarrow \quad C_{nfb}(P_x;\lambda,\kappa)<C_{fb}(P_x;\lambda,\kappa),
$$
where $x_P$ and $P_x$ are uniquely determined from each other via 
$$
\frac{\lambda(2\kappa+\lambda)}{\pi}\frac{x_P}{x_P^2+\kappa^2} - \frac{\lambda(2\kappa+\lambda)}{\pi\kappa}\arctan\left(\frac{x_P}{\kappa}\right) = P_x.
$$
Thus, it suffices to show that
\begin{equation}\label{proof-c0cf-Lambda1}
\Lambda_1(P;\lambda,\kappa) < 0 \quad\mbox{for all $P\ge P_{\max}$}
\end{equation}
and
\begin{equation}\label{proof-c0cf-Lambda2}
\Lambda_2(x;\lambda,\kappa)>0\quad\mbox{for all $x>0$}.
\end{equation}

To prove (\ref{proof-c0cf-Lambda1}) and (\ref{proof-c0cf-Lambda2}), first note that it can be easily verified that for any $P\ge P_{\max}$
\begin{equation*}
\begin{aligned}
\frac{d\Lambda_1(P;\lambda,\kappa)}{d P} &= \left(2\lambda+\frac{\lambda^2}{4\kappa}\right)P +2\lambda\kappa+\frac{3}{2}\lambda^2 +\frac{1}{\kappa}\lambda^3+\frac{1}{8\kappa^2}\lambda^4 <0,
\end{aligned}
\end{equation*}
and moreover, 
\begin{align*}
\Lambda_1(P_{\max};\lambda,\kappa) &=-\frac{\lambda^2(12\kappa^2+4\lambda\kappa-\lambda^2)}{8\kappa}<0,
\end{align*}
which, immediately implies (\ref{proof-c0cf-Lambda1}). On the other hand, it can be verified that for any $x>0,\lambda\in(-1,0)$ 
$$
\frac{d\Lambda_2(x;\lambda,1)}{dx}>0,
$$
which, together with the fact that $\Lambda_2(0;\lambda,1)=0$, immediately implies that $\Lambda_2(x;\lambda,1)>0$ for all $\lambda\in(-1,0), x>0$. Moreover, the definition of $\Lambda_2(x;\lambda,\kappa)$ implies that for any $\kappa>0,~ \lambda\in(-1,0),~ x>0,$
$$
\kappa\Lambda_2(x/\kappa;\lambda/\kappa,1)=\Lambda_2(x;\lambda,\kappa),
$$
which further leads to (\ref{proof-c0cf-Lambda2}), as desired.

Combining the above two cases, we conclude that $C_{nfb}(P;\lambda,\kappa)<C_{fb}(P;\lambda,\kappa)\leq 2C_{nfb}(P;\lambda,\kappa)$ and equality holds if and only if $\lambda=-\kappa$ and $P=\kappa/2$, as desired.
\end{proof}

Now, consider the discrete-time ACGN channel
\begin{equation}\label{discrete-time-GC}
Y_i = X_i +Z_i, \quad i=1,2,...,
\end{equation}
where $\{Z_i\}$ is a stationary Gaussian process. Denote by $C_{FB}(P)$ ({\em resp.,} $C_{NFB}(P)$) the feedback ({\em resp.,} non-feedback) capacity under average power constraint $P$. 
The following remark shows that the half-bit bound does not hold for continuous-time ACGN channels.

\begin{remark}
The half-bit bound by Cover and Pombra \cite{cover1989gaussian} states that for the channel (\ref{discrete-time-GC}),
$$
C_{FB}(P)\le C_{NFB}(P)+\frac{1}{2}\ln 2.
$$
A natural continuous-time version of the half-bit bound reads: for the channel (\ref{OU-WGN}), 
\begin{equation}\label{continuous-time-of-halfbit}
C_{fb}(P)\le C_{nfb}(P) +\frac{1}{2}\ln 2,
\end{equation}
which, however, is not true. To see this, let $\lambda=-\kappa,~\kappa>0$. It then follows from (\ref{proof-c0cf-k/2}) that 
$$
C_{fb}\left(\frac{\kappa}{2}\right)-C_{nfb}\left(\frac{\kappa}{2}\right)=\frac{\kappa}{2},
$$
which immediately disprove (\ref{continuous-time-of-halfbit}).
\end{remark}

The following remark provides a counterexample to the continuous-time version of Cover's $2P$ conjecture. 
\begin{remark}
Cover \cite{cover1987conjecture} conjectured that for the channel (\ref{discrete-time-GC}),
$$
C_{FB}(P)\le C_{NFB}(2P),
$$
which is often called Cover’s $2P$ conjecture and has been disproved by Kim \cite{kim2006counterexample}. A natural continuous-time version of Cover's $2P$ conjecture reads: for the channel (\ref{OU-WGN}),
\begin{equation}\label{continuous-time-of-Cover2P}
C_{fb}(P) \le C_{nfb}(2P),
\end{equation}
which, however, is not true. To see this, let $\lambda=-\kappa=-1, P=1$. Then, it follows from (\ref{proof-cocf-1}) and (\ref{proof-c0cf-cf1}) that 
$$
C_{fb}(1)\approx 1.4376 >1.25=C_{nfb}(2),
$$
which immediately disprove (\ref{continuous-time-of-Cover2P}).
\end{remark}

\section{Concluding Remarks}
In this paper, we studied the continuous-time ACGN channel in which the noise is a continuous-time ARMA(1,1) Gaussian process and find the feedback capacity of such a channel under average power constraint $P$ in closed form. 

It is noteworthy that the discrete-time approximation approach we used in the proof of the converse part (Lemma \ref{lemma:cft}) actually converts the continuous-time ARMA(1,1) Gaussian channel (\ref{OU-WGN}) to the sequence of discrete-time ARMA(1,1) channels (\ref{discrete-GC-nofeedback}) and (\ref{discrete-GC1-feedback}). Such a result seems to serve as a causality-preserving bridge connecting continuous-time Gaussian feedback channels and their associated discrete-time versions. However, the absence of similar results for general continuous-time ACGN channels in the presence of feedback has long remained a notable gap in the literature (see, e.g., \cite{han2021sampling,liu2019continuous}). Our results present a successful endeavor to address this gap and offer novel insights. 

Indeed, an immediate future work is to apply the framework and techniques developed in this paper to a broader range of continuous-time Gaussian channels in more general settings, detailed below. 
\begin{itemize}
\item First, we have shown that continuous-time SK coding scheme achieves the feedback capacity of the continuous-time ARMA(1,1) Gaussian channel for the case of $-2\kappa<\lambda<0$. This result connects the classical Schalkwijk–Kailath scheme (see, e.g., \cite{schalkwijk1966coding,schalkwijk1966coding1,schalkwijk1968center}) with the continuous-time one introduced in this paper. Interestingly, when $-2\kappa<\lambda<0$, the spectral density function of the continuous-time ARMA(1,1) Gaussian noise process is monotonously increasing on $[0,+\infty)$ and bounded above. A natural future work is to investigate the optimality of the continuous-time SK coding scheme for any continuous-time ACGN channel with such a noise spectrum. 
\item Second, Theorem \ref{thm:C0-Cf} reveals that feedback may not increase the capacity of the continuous-time ARMA(1,1) Gaussian channel even if the noise process is colored. It is worthwhile to investigate whether similar results can be established for more general continuous-time ACGN channels. 
\item Third, we believe that besides the continuous-time ARMA(1,1) Gaussian channel, our approach promises further applications in more general settings, which, for instance, may include a closed-form feedback capacity formula to continuous-time ACGN channels with a rational noise spectrum. The extensive studies on their discrete-time counterparts documented in the existing literature serve as a valuable foundation for pursuing these investigations.
\end{itemize}

\section*{Appendices}\appendix

\section{Proof of Theorem \ref{detailed-structure}}\label{lemma2}
It is known that $C_{fb,T}(P)$ can be achieved by transmitting a Gaussian message process $\Theta$ with an additive feedback term $\zeta$. Thus, it suffices to show that for each additive feedback coding scheme $(\Theta,X_\zeta)$ with $X_\zeta(t)=\Theta(t)-\zeta(t)$, where $\{\Theta(t)\}$ is Gaussian and $X_\zeta$ satisfies
$$
\frac{1}{T}\int_0^T\E[X_\zeta^2(t)]dt\le P,
$$ 
the Gaussian pair $(\Theta,X)$ of the form (\ref{opt-pair}) must satisfy
\begin{equation}\label{pf-3-4-01}
\frac{1}{T}\int_0^T\E[X^2(t)]dt\le P
\end{equation}
and have
\begin{equation}\label{pf-3-4-00}
I(\Theta_0^T;Y_0^T)=I(\Theta_0^T;Y_{\zeta,0}^T),
\end{equation}
where $Y_{\zeta}(t)=\int_0^t\Theta(s)ds-\int_0^t\zeta(s)ds+Z(t)$. Towards this goal, we first define
\begin{equation}\label{pf-3-4-1}
\widetilde{Y}(t) = \int_0^tL(t,u)dY_\zeta(u),
\end{equation}
where $L(t,u)$ is a Volterra kernel function satisfying $L(t,u)=1+\int_u^tl(s,u)ds$ for $t\ge u$. It then follows from (\ref{resolvent-kernel}) that
\begin{equation}\label{pf-3-4-2}
Y_\zeta(t)=\int_0^tH(t,u)d\widetilde{Y}(u),
\end{equation}
where  $H(t,u)$ is a Volterra kernel function satisfying $H(t,u)=1+\int_u^th(s,u)ds$ for $t\ge u$. 
Now, by \cite[Theorem 2]{ihara1980capacity}, $\zeta(t)$ can be of the form satisfying
$$
\int_0^t\zeta(s)ds=\int_0^t\int_0^sH(t,s)g(s,u)d\widetilde{Y}(u)ds,
$$
where $g(s,u)$ is a Volterra kernel function on $L^2([0,T]^2)$. Then, together with (\ref{pf-3-4-2}), we deduce that
\begin{align}
Y^\ast(t) &= Y_\zeta(t)+\int_0^t\int_0^sH(t,s)g(s,u)d\widetilde{Y}(u)ds\nonumber\\
&=\int_0^tH(t,u)d\widetilde{Y}(u)+\int_0^t\int_0^sH(t,s)g(s,u)d\widetilde{Y}(u)ds\nonumber\\
&=\int_0^tH(t,u)d\widetilde{Y}_g(u),\label{pf-3-4-3}
\end{align}
where 
\begin{equation}\label{pf-3-4-4}
\widetilde{Y}_g(t)=\int_0^tG(t,u)d\widetilde{Y}(u) 
\end{equation}
and $G(t,u)$ is a Volterra kernel function satisfying
$$
G(t,u)=1+\int_u^tg(s,u)ds,~~ t\ge u.
$$
Apparently, by (\ref{pf-3-4-1}) (\ref{pf-3-4-3}) (\ref{pf-3-4-4}) and (\ref{resolvent-kernel}), we have
\begin{equation*}
\mathcal{F}_t(Y^\ast)=\mathcal{F}_t(Y_\zeta),~~t\ge0.
\end{equation*}
which, together with the fact that $\zeta(t)$ is $\mathcal{F}_t(Y_\zeta)$-measurable, immediately implies that
\begin{align*}
\int_0^T\E[X^2(t)]dt&=\int_0^T\E\left[\vert\Theta(t)-\E[\Theta(t)\vert\mathcal{F}_t(Y^\ast)]\vert^2\right]dt\\
&=\int_0^T\E\left[\vert\Theta(t)-\E[\Theta(t)\vert\mathcal{F}_t(Y_{\zeta})] \vert^2\right]dt\\
&\le\int_0^T\E[\vert\Theta(t)-\zeta(t)   \vert^2]dt \\
&\le PT,
\end{align*}
establishing (\ref{pf-3-4-01}).
Moreover, note that $\E[\Theta(t)\vert \mathcal{F}_t(Y^\ast)]$ can be written as (see, e.g., \cite{ihara1980capacity})
\begin{equation*}\label{Appendix-B-1}
\E[\Theta(t)\vert \mathcal{F}_t(Y^\ast)]=\int_0^tf(t,s)dY^\ast(s),
\end{equation*}
where $f(t,s)$ is a Volterra kernel function on $L^2([0,T]^2)$. Therefore, we have 
\begin{equation*}\label{Appendix-B-2}
\begin{aligned}
Y(t)&=\int_0^t\Theta(s)ds-\int_0^t\E[\Theta(s)\vert \mathcal{F}_s(Y^\ast)]ds+Z(t)\\
&=Y^\ast(t)-\int_0^t\int_0^sf(s,u)dY^\ast(u)ds\\
&=\int_0^tF(t,u)dY^\ast(u),
\end{aligned}
\end{equation*}
where $F(t,u)$ is a Volterra kernel function satisfying $F(t,u)=1-\int_u^tf(s,u)ds$ for $t\ge u$. It then immediately follows that  
$$
\mathcal{F}_t(Y^\ast)=\mathcal{F}_t(Y),~~t\ge0,
$$
and so the pair $(\Theta,X)$ characterizes an additive feedback coding scheme (see also \cite[Lemma 5]{ihara1980capacity}). Therefore, it follows from Lemma \ref{mutual-filter} that
\begin{align*}
I(\Theta_0^T;Y_0^T)=I(\Theta_0^T;Y_0^{\ast,T})=I(\Theta_0^T;Y_{\zeta,0}^T),
\end{align*}
establishing (\ref{pf-3-4-00}). The proof is then complete.

\section{Proof of Lemma \ref{exist-limit-ode}}\label{appendix-proof-of-ode}
We first prove the existence and uniqueness of the solution $g(t)$. Let $P_3(y;t,P)$ denote the polynomial (in $y$): $-Py^3+P/\sqrt{2}y^2+p(t)y+q(t)/\sqrt{2}$. Since $p(t),q(t)$ are continuous, an application of \cite[Theorem (7.6)]{amann2011ordinary} gives rise to a unique nonextendible solution $g(t)$, which is either defined for all $t\ge 0$ or blows up at some $t>0$. In fact, the domain of $g(t)$ extends to the infinity since it cannot blow up in finite interval. Indeed, by way of contradiction, suppose that there exists $T_0<\infty$ such that
\begin{equation}\label{blow-up}
\lim_{t\to T_0^-}g(t)=+\infty.
\end{equation}
Then, it follows from the continuity of $p(t),q(t)$ at $T_0$ that there exists $\epsilon>0$ such that $P_3(g(t);t,P)<0$ for all $t\ge T_0-\epsilon$. However, by (\ref{ode}), it holds that $g^\prime(t)<0$ for $t\ge T_0-\epsilon$, which contradicts (\ref{blow-up}), as desired.

Next, we shall prove the ``moreover'' part. To achieve this, let 
$$
P_3(y;P)\triangleq-Py^3+\frac{P}{\sqrt{2}}y^2+p y+\frac{q}{\sqrt{2}}.
$$
Since
\begin{equation}\label{p-q-t}
\lim_{t\to\infty}p(t)=p,~~~\lim_{t\to\infty}q(t)=q,
\end{equation}
we have that $\lim_{t\to\infty}P_3(y;t,P)=P_3(y;P)$.

Next, we deal with the following three cases:
\begin{itemize}
\item[(\uppercase\expandafter{\romannumeral1})] The cubic $P_3(y;P)$ has one real root $(r_{11})$ and two non-real complex conjugate roots $(r_{12}=\bar{r}_{13})$;
\item[(\uppercase\expandafter{\romannumeral2})] The cubic $P_3(y;P)$ has three distinct real roots ($r_{21}<r_{22}<r_{23}$);
\item[(\uppercase\expandafter{\romannumeral3})] The cubic $P_3(y;P)$ has a simple root $(r_{31})$ and a double root $(r_{32}=r_{33}$).
\end{itemize}
We shall prove that the solution $g(t)$ converges to some real root $r_{ij}$ as $t\to\infty$ for $i,j\in\{1,2,3 \}$ case by case. For $x\in\mathbb{R}, M>0$, let $B(x,M)\triangleq (x-M,x+M)$.

\noindent \textbf{Case (\uppercase\expandafter{\romannumeral1})}. Let $\varepsilon>0$ be a sufficiently small constant. It then follows immediately from the continuity of roots of polynomial (see, e.g., \cite[Theorem B]{harris1987shorter}) and (\ref{p-q-t}) that there exists $T_\varepsilon>0$ such that for any $t\ge T_\varepsilon$, $P_3(y;t,P)$ admits the unique real root $r_{11}(t)$ satisfying
\begin{equation}\label{f-epsilon}
\vert r_{11}(t)-r_{11}\vert < \varepsilon.
\end{equation}
It then remains to show that there exists $T^\ast_\varepsilon\ge T_\varepsilon$ such that
\begin{equation}\label{f-g}
\sup_{t\ge T_\varepsilon^{\ast}} \vert g(t)-r_{11}(t)\vert\le2\varepsilon.
\end{equation}
Indeed, it immediately follows from (\ref{f-epsilon}) and (\ref{f-g}) that $\lim_{t\to\infty}g(t)=r_{11}$, as desired.

Note that by (\ref{ode}), we have
\begin{equation}\label{f-g-relation}
\begin{aligned}
 g^\prime(t)=0 ~~~&\Longleftrightarrow~~~ g(t)=r_{11}(t);\\
 g^\prime(t)>0 ~~~&\Longleftrightarrow~~~ g(t)<r_{11}(t);\\
 g^\prime(t)<0 ~~~&\Longleftrightarrow~~~ g(t)>r_{11}(t).
\end{aligned}
\end{equation}
Clearly, if $g(T_\varepsilon)\in \overline{B(r_{11},\varepsilon)}$, then (\ref{f-g}) holds true with $T_\varepsilon^\ast=T_\varepsilon$. WLOG, we assume in the following that 
$g(T_{\varepsilon})> r_{11}+\varepsilon$ since the proof is similar if 
$g(T_{\varepsilon})< r_{11}-\varepsilon$. We now claim that there exists $t^\ast\ge T_\varepsilon$ such that $g(t^\ast)\in \overline{B(r_{11},\varepsilon)}$.
To see this, by way of contradiction, we suppose the opposite is true, that is, 
\begin{equation}\label{r11-contra}
g(t)\notin \overline{B(r_{11},\varepsilon)}\quad\mbox{ for all $t \ge T_{\varepsilon}$}.
\end{equation}
It then follows from (\ref{f-g-relation}) that for any $t\ge T_{\varepsilon}$
\begin{equation*}
g^\prime(t)<0 ~~~\text{and}~~~g(t)>r_{11}(t),
\end{equation*}
which, together with (\ref{f-epsilon}), implies that
$$
g(t)\ge \lim_{t\to\infty}g(t)\ge \lim_{t\to\infty}r_{11}(t)=r_{11}.
$$
Hence, both $\lim_{t\to\infty}g(t)$ and $\lim_{t\to\infty}g^\prime(t)$ exist, which implies $\lim_{t\to\infty}g^\prime(t)$ $=0$. Then, by (\ref{ode}), we have $\lim_{t\to\infty}g(t)=r_{11}$, which contradicts (\ref{r11-contra}). Consequently, (\ref{f-g}) immediately follows from (\ref{f-g-relation}) with $T_\varepsilon^\ast=t^\ast$, as desired.

\noindent\textbf{Case (\uppercase\expandafter{\romannumeral2})}. The proof of this case is largely similar to that of Case (\uppercase\expandafter{\romannumeral1}), except that $g(t)$ may converge to the middle root $r_{22}$ as $t\to\infty$. Indeed, pick $\varepsilon>0$ such that $B(r_{2i},\varepsilon)\cap B(r_{2j},\varepsilon)=\varnothing$ for $i\neq j$.
Then, there exists $T_\varepsilon>0$ such that the polynomial $P_3(y;t,P)$ admits three real roots $\{r_{2j}(t),j=1,2,3\}$ satisfying $r_{2j}(t)\in B(r_{2j},\varepsilon),j=1,2,3,$ for all $t\ge T_\varepsilon$. 
On the one hand, the same argument in Case (\uppercase\expandafter{\romannumeral1}) yields $\lim_{t\to\infty}g(t)=r_{21}$ if $g(T_{\varepsilon})\in (-\infty, r_{22}-\varepsilon) $ 
or $\lim_{t\to\infty}g(t)=r_{23}$ if $g(T_\varepsilon)\in (r_{22}+\varepsilon,+\infty)$. 
On the other hand, if $g(T_\varepsilon)\in\overline{B(r_{22},\varepsilon)}$, then there will be only two subcases for $\{g(t),t\ge T_\varepsilon\}$, i.e., either $g(t)\in\overline{B(r_{22},\varepsilon)}$ for all $t\ge T_\varepsilon$ or $g(t^\prime)\notin\overline{B(r_{22},\varepsilon)}$ for some $t^\prime>T_\varepsilon$. The latter subcase can be proved similarly as done before. For the former subcase, we have $\lim_{t\to\infty}g(t)=r_{22}$, as desired.

\noindent\textbf{Case (\uppercase\expandafter{\romannumeral3})}. WLOG, we assume that $r_{31}<r_{32}=r_{33}$. Let $\varepsilon>0$ be given such that $B(r_{31},\varepsilon)\cap B(r_{3j},\varepsilon)=\varnothing$ for $j=2,3$. 
As in Case (\uppercase\expandafter{\romannumeral2}), it suffices to consider the subcase $g(T_\varepsilon)\in \overline{B(r_{32},\varepsilon)}$. By (\ref{f-g-relation}), there are subcases for $\{g(t),t\ge T_\varepsilon\}$, i.e., either $g(t)\in \overline{B(r_{32},\varepsilon)}$ for all $t\ge T_\varepsilon$ or $g(t)\in (r_{31}+\varepsilon,r_{32}-\varepsilon)$ for some $t=T^\ast_\varepsilon\ge T_\varepsilon$. The former subcase leads to $\lim_{t\to\infty}g(t)=r_{32}$ and the latter subcase of $g(t)$ converges to $r_{31}$. The proof is then complete.

\section{Proofs of $(a)-(f)$}\label{proof-of-cf}
We shall first give the proofs of $(a)$, $(c)$ and $(e)$ as follows.
\begin{proof}[Proof of (a)]
The equality (a) follows from Lemma \ref{mutual-filter}.
\end{proof}
\begin{proof}[Proof of (c)] 
It is easy to show that $Y^{(n)}_{k}$ is a linear combination of $Y^{\ast,(n)}_{i},i=0,1,...,k$, and vice versa, which implies (c).
\end{proof}
\begin{proof}[Proof of (e)]
From the stationarity of the ARMA(1,1) process (\ref{arma11-stationary}), it follows that $C_{FB,n}$ is super-additive \cite{kim2009feedback}:
$$
mC_{FB,m}+nC_{FB,n}\le (m+n)C_{FB,m+n}\quad \mbox{for all $m,n=1,2,...$}.
$$
As a consequence, $C_{FB,n}(P)\le C_{FB}(P)$ for all $n\in\mathbb{N}$, which implies $(e)$.
\end{proof}

We are now in a position to give the proofs of $(b)$ and $(d)$.

\begin{proof}[Proof of (b)]
Let $\widetilde{\Theta}(t)=\int_0^t \Theta(s)ds$ for $t\in[0,T]$. Define $\widetilde{\Theta}^{(n)}=\{\widetilde{\Theta}^{(n)}(t);0\le t\le T\}$ and $Y^{\ast,(n)}=\left\{Y^{\ast,(n)}(t);0\le t\le T\right\}$ as follows:
\begin{equation*}
\begin{aligned}
\widetilde{\Theta}^{(n)}(t)=\int_0^{\Delta_n(t)}\Theta(s)ds,\quad
Y^{\ast,(n)}(t)=\sum_{i=0}^{N_{\Delta_n}(t)-1}Y^{\ast,(n)}_{i},
\end{aligned}
\end{equation*}
where $\Delta_n(t)\triangleq \max \{t^{(n)}_k\vert t\ge t^{(n)}_k,k\le n, k\in \mathbb{N} \}$ and $N_{\Delta_n}(t)\triangleq \max \{k\vert t\ge t^{(n)}_k,k\le n, k\in \mathbb{N} \}$. 
We further define an approximation process $\{\hat{Z}^{(n)}_\dag(t) \}$ of $Z_\dag$ as
\begin{equation*}
\hat{Z}^{(n)}_\dag(t)=e^{-\kappa t}\sum_{i=0}^{k-1}e^{\kappa t_{i+1}}B^{(n)}_i,~~~ \text{if}~t\in[t^{(n)}_k,t^{(n)}_{k+1}).
\end{equation*}
Let $\hat{Z}_{\dag,k}^{(n)}\triangleq \hat{Z}^{(n)}_\dag(t^{(n)}_{k+1})-\hat{Z}^{(n)}_\dag(t^{(n)}_k)$. Then, we have
\begin{equation}\label{Z0-hat-n}
\begin{aligned}
\hat{Z}_{\dag,k}^{(n)}&=B^{(n)}_k-\kappa d^{(n)}_k\sum_{i=0}^{k-1}e^{\kappa t_{i+1}}B^{(n)}_i,  \\
&=B^{(n)}_k-\int_{t^{(n)}_k}^{t^{(n)}_{k+1}}\hat{Z}^{(n)}_\dag(s)ds.
\end{aligned}
\end{equation}
Hence, $\{\widetilde{Z}_k^{(n)}\}$ defined in (\ref{z-app-stat}) can be equivalently written as
\begin{equation}\label{Z-equal-expression}
\widetilde{Z}^{(n)}_k=\left(1+\frac{\lambda}{\kappa}\right)B^{(n)}_k-\frac{\lambda}{\kappa}\hat{Z}^{(n)}_{\dag,k}+\lambda d^{(n)}_k m(\delta_n)\zeta_0.
\end{equation}
It follows from (\ref{Z0-Zstar}) and (\ref{Zs-Z0-relation}) that $\{Z(t) \}$ can be similarly expressed as
\begin{equation*}
\begin{aligned}
Z(t)&=Z_\ast(t)+\lambda\zeta_0\int_0^te^{-\kappa s}ds,\\
&=\left(1+\frac{\lambda}{\kappa}\right)B(t)-\frac{\lambda}{\kappa}Z_\dag(t)+\lambda\zeta_0\int_0^te^{-\kappa s}ds,~~~t\in[0,T].
\end{aligned}
\end{equation*}
Therefore, for any $t\in[0,T]$, we have 
\begin{equation}
\begin{aligned}
Y^{\ast,(n)}(t)&=\sum_{i=0}^{N_{\Delta_n}(t)-1}\left(\Theta^{(n)}_i+\widetilde{Z}^{(n)}_i \right)     \\
&= \int_0^{\Delta_n(t)}\Theta(s)ds+ \left(1+\frac{\lambda}{\kappa}\right)B(\Delta_n(t)) -\frac{\lambda}{\kappa}\hat{Z}^{(n)}_\dag(\Delta_n(t))+\lambda\zeta_0 m(\delta_n)\int_0^{\Delta_n(t)}e^{-\kappa s}ds.
\end{aligned}
\end{equation}
Hence, we can readily prove that $\{\widetilde{\Theta}^{(n)},Y^{\ast,(n)}\}$ converges in distribution to $\{\widetilde{\Theta},Y^\ast\}$. By the lower semi-continuity of mutual information \cite[p. 211-212]{gelfand1959calculation}, we obtain 
$$
I(\widetilde{\Theta}_0^T;Y_0^{\ast,T})\le \varliminf_{n\to\infty}I(\widetilde{\Theta}_0^{(n),T};Y_0^{\ast,(n),T}).
$$
This, together with $I(\Theta_0^T;Y_0^{\ast,T})=I(\widetilde{\Theta}_0^T;Y_0^{\ast,T})$ and $I(\widetilde{\Theta}_0^{(n),T};Y_0^{\ast,(n),T})=I(\{\Theta_k^{(n)}\};\{Y_k^{\ast,(n)}\})$, implies (b).
\end{proof}

\begin{proof}[Proof of (d)]
Recall that we have constructed an $n$-block discrete-time ARMA$(1,1)$ Gaussian channel with feedback
\begin{equation*}\label{induced-arma-gc}
\frac{Y^{(n)}_{k}}{\sqrt{\delta_n}}=\frac{\Theta^{(n)}_k-\zeta^{(n)}_k}{\sqrt{\delta_n}}+\frac{\widetilde{Z}^{(n)}_{k}}{\sqrt{\delta_n}},~~~k=0,1,..,n-1.
\end{equation*}
The energy $E(\delta_n)$ and average power $P(\delta_n)$ for such a channel can be computed as
\begin{align*}
E(\delta_n)=\sum_{k=0}^{n-1}\E\left[\left\vert\frac{\Theta^{(n)}_k-\zeta^{(n)}_k}{\sqrt{\delta_n}} \right\vert^2       \right],\quad
P(\delta_n)=\frac{1}{n}E(\delta_n).
\end{align*}
Define a Volterra kernel $K^{(n)}(t,s)\in L^2([0,T]^2)$ by 
$$
K^{(n)}(t,s)=\left\{\begin{aligned}
&K(t^{(n)}_k,t^{(n)}_l),\quad \text{if} ~(t,s)\in [t^{(n)}_k,t^{(n)}_{k+1})\times [t^{(n)}_l,t^{(n)}_{l+1}),l< k;\\
&0,~~~~~~~~~~~~~~\text{otherwise}
\end{aligned}
\right.
$$
and a random process $\{\zeta^{(n)}(t)\}$ by
$$
\zeta^{(n)}(t)=\sum_{i=0}^{N_{\Delta_n}(t)-1}K(\Delta_n(t),t^{(n)}_i)Y^{\ast,(n)}_{i}
$$
respectively. By the assumption (C.0), we have that $\lim_{n\to\infty}\| K^{(n)}-K \|_2=0$, where $\|\cdot \|_2$ denotes the usual norm on $L^2([0,T]^2)$. Thus, it is clear from (\ref{Ystar}) and (\ref{Z-equal-expression}) that
\begin{equation}\label{convergence-1}
\varlimsup_{n\to\infty}\int_0^T\E\left[\left\vert \int_0^t K^{(n)}(t,s)dY^{\ast}(s)-\int_0^tK(t,s)dY^\ast(s) \right\vert^2\right]dt\le \varlimsup_{n\to\infty}\| K^{(n)}-K\|_2=0.
\end{equation}
Furthermore, set $\Delta Y^\ast_i=Y^\ast(t^{(n)}_{i+1})-Y^\ast(t^{(n)}_{i})$ and $\Delta Z_{\dag,i}=Z_\dag(t^{(n)}_{i+1})-Z_\dag(t^{(n)}_{i})$ respectively. It then follows from (\ref{Z0-hat-n}) and (\ref{z-app-stat}) that 
\begin{equation}\label{Z-diff}
\E[\vert\hat{Z}^{(n)}_{\dag,i}-\Delta Z_{\dag,i} \vert^2]\le 2\kappa^2 \vert d^{(n)}_k\vert^2\delta_n \frac{(e^{\kappa\delta_n}-1)^2(e^{2\kappa t^{(n)}_k}-1)}{e^{2\kappa \delta_n}-1}+\frac{1}{2}\delta_n e^{-2\kappa\delta_n}(e^{2\kappa\delta_n}-1)^2
\end{equation}
and
\begin{align}\label{Y-ast-diff}
\E[\vert &Y^{\ast,(n)}_i-\Delta Y^\ast_i \vert^2]=\frac{\lambda^2}{\kappa^2}\E[\vert\hat{Z}^{(n)}_{\dag,i}-\Delta Z_{\dag,i} \vert^2]+\frac{\lambda^2}{2\kappa}\vert d^{(n)}_i\vert^2(1-m(\delta_n))^2
\end{align}
for all $i$. Therefore, we have 
\begin{equation}\label{convergence-2}
\begin{aligned}
\lim_{n\to\infty}\int_0^T\E[\vert &\zeta^{(n)}(t)-\int_0^tK^{(n)}(t,s)dY^{\ast}(s)    \vert^2]dt \\ &=\lim_{n\to\infty}\int_0^T\E[\vert \sum_{i=0}^{N_{\Delta_n}(t)-1}K^{(n)}(\Delta_n(t),t^{(n)}_i)(Y^{\ast,(n)}_i-\Delta Y^\ast_i) \vert^2]dt\\
&\le \lim_{n\to\infty}\int^T_0 \sum_{i=0}^{N_{\Delta_n}(t)-1}\vert K^{(n)}(\Delta_n(t),t^{(n)}_i)\vert^2\sum_{i=0}^{N_{\Delta_n}(t)-1}\E[\vert Y^{\ast,(n)}_i-\Delta Y^\ast_i \vert^2]dt\\
&\overset{(a)}{\le}  \int^T_0\lim_{n\to\infty} \sum_{i=0}^{N_{\Delta_n}(t)-1}\vert K^{(n)}(\Delta_n(t),t^{(n)}_i)\vert^2\delta_n\sum_{i=0}^{N_{\Delta_n}(t)-1}\frac{1}{\delta_n}\E[\vert Y^{\ast,(n)}_i-\Delta Y^\ast_i \vert^2]dt \\
&\overset{(b)}{=}0,
\end{aligned}
\end{equation}
where (a) follows from the general Lebesgue dominated convergence theorem, and where in (b) we have used the result derived from the assumption (C.0) that
$$
\lim_{n\to\infty} \sum_{i=0}^{N_{\Delta_n}(t)-1}\vert K^{(n)}(\Delta_n(t),t^{(n)}_i)\vert^2\delta_n =\int^t_0 K^2(t,s)ds~~\text{for all}~t,
$$
and another result derived from (\ref{Z-diff}) and (\ref{Y-ast-diff}) that
\begin{align*}
\lim_{n\to\infty}\sum_{i=0}^{N_{\Delta_n}(t)-1}\frac{1}{\delta_n}\E[\vert &Y^{\ast,(n)}_i-\Delta Y^\ast_i \vert^2]=0.
\end{align*}
Thus, we conclude that
\begin{align}
\lim_{n\to\infty}\int_0^T\E[\vert\Theta(t)-\zeta^{(n)}(t)\vert^2]dt=\int_0^T\E\left(\Theta(t)-\int_0^sK(t,s)dY^\ast(s)   \right)^2dt,\label{121}
\end{align}
which follows from (\ref{convergence-1}), (\ref{convergence-2}) and 
\begin{align*}
\int_0^T\E[\vert\zeta^{(n)}(t)-\int_0^tK(t,s)&dY^\ast(s) \vert^2]dt\le \int_0^T2\E[\vert\zeta^{(n)}(t)-\int_0^tK^{(n)}(t,s)dY^\ast(s)\vert^2]dt\\
&+\int_0^T2\E\left[\left\vert\int_0^tK^{(n)}(t,s)dY^\ast(s)-\int_0^tK(t,s)dY^\ast(s) \right\vert^2\right]dt.
\end{align*}
Now, by H\"older's inequality, we have
\begin{align}
E(\delta_n)&\le \sum_{k=0}^{n-1}\int_{t_k^{(n)}}^{t^{(n)}_{k+1}}\E\left(\Theta(t)-\sum_{i=0}^{k-1}K(t^{(n)}_k,t^{(n)}_i)Y^{\ast,(n)}_{i}\right)^2dt\nonumber\\
&=\int_0^T\E[\vert\Theta(t)-\zeta^{(n)}(t)\vert^2]dt,\nonumber
\end{align}
which, together with (\ref{121}), implies that there exists an error function $e(\delta_n)$ such that
$$
\lim_{n\to\infty}e(\delta_n)=0
$$
and
\begin{equation*}
P(\delta_n)\le P\delta_n +\frac{e(\delta_n)}{n}.
\end{equation*}
Then, (d) immediately follows from the definition of $n$-block capacity.

\end{proof}

Note that $\{\widetilde{Z}^{(n)}_{k+1}/\sqrt{\delta_n},k=0,1,...,n-1\}$ defined by (\ref{z-app-stat}) satisfies
$$
\frac{\widetilde{Z}^{(n)}_{k+1}}{\sqrt{\delta_n}}+\phi(\delta_n)\frac{\widetilde{Z}^{(n)}_{k}}{\sqrt{\delta_n}}=\frac{B^{(n)}_{k+1}}{\sqrt{\delta_n}}+\theta(\delta_n)\frac{B^{(n)}_k}{\sqrt{\delta_n}},
$$
where $\phi(\delta_n)=-e^{-\kappa \delta_n}$ and $\theta(\delta_n)=\lambda/\kappa -(\lambda/\kappa+1)e^{-\kappa \delta_n}$. 
We are now in a position to give the proof of $(f)$.
\begin{proof}[Proof of (f)]
In the following, we deal with the case $\lambda/\kappa\ge -1$ only, since the case $\lambda/\kappa<-1$ can be proved in a parallel manner. 

First of all, it is clear that
\begin{equation*}
\begin{aligned}
\text{sgn}(\phi(\delta_n)-\theta(\delta_n))=\text{sgn}\left(\frac{\lambda}{\kappa}(e^{-\kappa\delta_n}-1)\right)=\left\{\begin{aligned}
-1, &~~~\text{if}~~\lambda>0;\\
0, & ~~~\text{if}~~\lambda=0;\\
1, &~~~\text{if}~~\lambda<0.
\end{aligned}
\right.
\end{aligned}
\end{equation*}
Next, we complete the proof by considering the following three cases:\\
\textbf{Case 1}: $-\kappa\le\lambda<0$. For any arbitrarily small $\epsilon>0$ there exists a sufficiently large $N$ such that for $n\ge N$, $P\delta_n+e(\delta_n)/n\le (P+\epsilon) \delta_n\triangleq P_{\delta_n}(\epsilon)$ and $\vert\theta(\delta_n)\vert\le 1$. Thus, by Theorem \ref{kim}, we obtain $C_{FB}(P_{\delta_n}(\epsilon))=-\log x(\delta_n)$, where $x(\delta_n)$ is the unique positive root of the following polynomial
\begin{equation} \label{Kim's-1}
P_{\delta_n}(\epsilon) x^2 = \frac{(1-x^2)(1+\theta(\delta_n)x)^2}{(1+\phi(\delta_n)x)^2}.
\end{equation}

By the continuity of roots of polynomial, we infer that $\lim_{\delta_n\to 0^+} x(\delta_n)=1$. Moreover, by elementary calculus, it holds that $x(\delta_n)$ is differentiable in $\delta_n$ over $(0,\delta_N)$ and $\lim_{\delta_n\to 0^+} x^\prime(\delta_n)$ exists. Since
\begin{align*}
\lim_{n\to\infty}\frac{C_{FB}(P_{\delta_n}(\epsilon))}{\delta_n}&=\lim_{n\to\infty}\frac{-\log x(\delta_n)}{\delta_n}\\
&=\lim_{\delta_n\to 0^+}-\frac{x^\prime(\delta_n)}{x(\delta_n)}\\
&=\lim_{\delta_n\to 0^+} -x^\prime(\delta_n),
\end{align*}
$\lim_{n\to\infty}C_{FB}(P_{\delta_n}(\epsilon))/\delta_n$ exists, which is denoted by $\beta_{P+\epsilon}$. Then, it holds that
\begin{equation}\label{proof-f-x_delta_n}
\frac{1}{x(\delta_n)}=\beta_{P+\epsilon}\delta_n +1 +o(\delta_n)
\end{equation}
for $n$ large enough. Now, substituting (\ref{proof-f-x_delta_n}) into (\ref{Kim's-1}) and letting $n\to\infty$, we establish the equation 
\begin{equation*}
(P+\epsilon)(\beta_{P+\epsilon}+\kappa)^2=2\beta_{P+\epsilon}^2(\beta_{P+\epsilon}+\kappa+\lambda)^2.
\end{equation*}
Thus, we have
\begin{align*}
\varliminf_{n\to\infty}\frac{1}{\delta_n}C_{FB}(P\delta_n+\frac{e(\delta_n)}{n})&\le \lim_{n\to\infty}\frac{C_{FB}(P_{\delta_n}(\epsilon))}{\delta_n}\\
&=\beta_{P+\epsilon}.
\end{align*}
Letting $\epsilon\to 0$, we conclude $\lim_{\epsilon\to 0^+}\beta_{P+\epsilon}=x_0(P;\lambda,\kappa)$. Thus, we complete the proof of $(f)$ in this case. \\
\textbf{Case 2}: $\lambda>0$. By Theorem \ref{kim} again, the polynomial equation (\ref{Kim's-1}) in Case 1 becomes
$$
P_{\delta_n}(\epsilon) x^2 = \frac{(1-x^2)(1-\theta(\delta_n)x)^2}{(1-\phi(\delta_n)x)^2}.
$$
Similarly, we can also obtain that $\beta_{P+\epsilon}=(P+\epsilon)/2$, as desired. \\
\textbf{Case 3}: $\lambda=0$. In this case, the continuous-time ARMA$(1,1)$ Gaussian channel (\ref{ou-awgn-main}) boils down to the continuous-time AWGN channel (\ref{AWGN2}). Indeed, similarly as above, we can readily show that $\beta_{P+\epsilon}=(P+\epsilon)/2$, which is our desired result. 
\end{proof}

\textbf{Acknowledgement.} The work of G. Han has been supported by the Research Grants Council of the Hong Kong Special Administrative Region, China, under Project 17304121 and by the National Natural Science Foundation of China, under Project 61871343. The work of S. Shamai has been supported by the German Research Foundation (DFG) via the German-Israeli Project Cooperation (DIP), under Project SH $1937/1$-$1$.

\bibliographystyle{ieeetr}
\bibliography{OU_Colored}

\end{document}